\documentclass[10pt,sigconf]{acmart}

\AtBeginDocument{%
	\providecommand\BibTeX{{%
			\normalfont B\kern-0.5em{\scshape i\kern-0.25em b}\kern-0.8em\TeX}}}

\usepackage[utf8]{inputenc}
\usepackage{graphicx}
\usepackage{balance}
\usepackage{colortbl}
\usepackage{amsmath}
\usepackage{amsfonts}
\usepackage{amssymb}
\usepackage{mathtools}
\usepackage{tensor}
\usepackage{listings}
\usepackage[framemethod=TikZ]{mdframed}
\usepackage{pgfplots}
\usepackage{pgfplotstable}
\usepackage{textcomp}
\usepackage{tikz}
\usepackage{fancyvrb}
\usepackage{multicol}
\usepackage{diagbox}
\usepackage{subcaption}
\usepackage{xcolor}
\usepackage[algo2e,ruled,vlined,linesnumbered]{algorithm2e}

\setcopyright{acmcopyright}
\copyrightyear{2019}
\acmYear{2019}
\acmDOI{10.1145/1122445.1122456}
\acmConference[SIGMOD '20]{ACM SIGMOD International Conference on Management of Data}
{June 14--19, 2020}{Portland, OR}

\newcommand{\mq}[1]{\text{'#1'}}

\usetikzlibrary{arrows,decorations.pathmorphing,positioning}
\usetikzlibrary{shapes.geometric}
\usetikzlibrary{calc}
\usetikzlibrary{fit}
\usetikzlibrary{automata,positioning}
\usetikzlibrary{intersections}
\usetikzlibrary{patterns}
\usetikzlibrary{arrows.meta,bending}
\tikzset{
	>=stealth',
	help lines/.style={dashed, thick},
	axis/.style={<->},
	important line/.style={thick},
	connection/.style={thick, dotted},
	vertex/.style = {circle, fill = black, outer sep = 1pt, inner sep = 0.5pt}
}
\tikzstyle{newnode} = [rectangle, rounded corners, text centered, draw=black, align=center]
\tikzstyle{arrow} = [thick,->,>=stealth]
\tikzstyle{level 1}=[level distance=3.5cm, sibling distance=3.5cm]
\tikzstyle{level 2}=[level distance=3.5cm, sibling distance=2cm]
\tikzstyle{bag} = [text width=2em, text centered]
\tikzstyle{end} = [circle, minimum width=3pt,fill, inner sep=0pt]
\tikzset{every picture/.style={remember picture}}

\pgfdeclarelayer{background}
\pgfsetlayers{background,main}

\pgfplotsset{table/search path={data}}

\newcolumntype{a}{>{\columncolor[gray]{0.8}}c}
\newcolumntype{m}{>{\columncolor[gray]{0.5}}c}

\newcommand*\ellipsed[1]{
	\tikz[baseline=(char.base)]{\node[shape=ellipse,draw,inner sep=.3pt] (char) {#1};}}
\newcommand*\rectangled[1]{
	\tikz[baseline=(char.base)]{\node[shape=rectangle,draw,inner sep=1pt] (char) {#1};}}

\def\mconc{\,{\scriptstyle\Box}\,}

\def\tra{\mathtt{tra}}
\def\mmu{\mathtt{mmu}}
\def\emu{\mathtt{emu}}
\def\opd{\mathtt{opd}}
\def\cpd{\mathtt{cpd}}
\def\qqr{\mathtt{qqr}}
\def\rqr{\mathtt{rqr}}
\def\inv{\mathtt{inv}}
\def\add{\mathtt{add}}
\def\sub{\mathtt{sub}}
\def\rnk{\mathtt{rnk}}
\def\det{\mathtt{det}}
\def\sol{\mathtt{sol}}
\def\evc{\mathtt{evc}}
\def\evl{\mathtt{evl}}

\def\chf{\mathtt{chf}}
\def\dsd{\mathtt{dsv}}
\def\usd{\mathtt{usv}}
\def\vsd{\mathtt{vsv}}

\def\op{\mathtt{op}}

\def\invMA{\mathtt{INV}}
\def\addMA{\mathtt{ADD}}
\def\subMA{\mathtt{SUB}}

\def\traMA{\mathtt{TRA}}
\def\mmuMA{\mathtt{MMU}}
\def\emuMA{\mathtt{EMU}}
\def\opdMA{\mathtt{OPD}}
\def\cpdMA{\mathtt{CPD}}
\def\qqrMA{\mathtt{QQR}}
\def\rqrMA{\mathtt{RQR}}
\def\rnkMA{\mathtt{RNK}}
\def\detMA{\mathtt{DET}}
\def\solMA{\mathtt{SOL}}
\def\traMA{\mathtt{TRA}}
\def\evcMA{\mathtt{EVC}}
\def\evlMA{\mathtt{EVL}}
\def\dsdMA{\mathtt{DSV}}
\def\usdMA{\mathtt{USV}}
\def\vsdMA{\mathtt{VSV}}
\def\chfMA{\mathtt{CHF}}

\def\opMA{\mathtt{OP}}

\newcommand{\opr}[3]{\hspace{-0pt}\tensor[_{#3}]{\mathtt{#1}}{_{\hspace{-0.5pt}#2}}}
\newcommand{\oprc}[4]{\hspace{-0pt}\tensor*[_{#3}]{\mathtt{#1}}{_{#2}}}

\makeatletter
\newcommand{\lstuppercase}
{\uppercase\expandafter{\expandafter\lst@token\expandafter{\the\lst@token}}}
\newcommand{\lstlowercase}
{\lowercase\expandafter{\expandafter\lst@token\expandafter{\the\lst@token}}}
\makeatother

\lstdefinestyle{SQLstyle}{%
	language=SQL,%
	basicstyle=\ttfamily\color{black}\fontsize{8}{9}\selectfont,%
	deletekeywords={ON, BY},
	otherkeywords={with, WITH},%
	morekeywords={with, WITH},%
	morekeywords=[2]{mmu, inv, cpd},
	morekeywords=[3]{on},
	morekeywords=[4]{by},
	keywordstyle=\ttfamily\lstuppercase\bfseries\color{black},%
	keywordstyle=[2]\ttfamily\lstuppercase\bfseries\color{black!66!green},
	keywordstyle=[3]\ttfamily\lstuppercase\bfseries\color{black!66!green!80!white},
	keywordstyle=[4]\ttfamily\lstuppercase\bfseries\color{black!66!green!60!white},
	identifierstyle=\ttfamily\color{black},%
	stringstyle=\ttfamily\color{black},%
	showstringspaces=false,%
	mathescape=true,%
	upquote=true,%
}

\lstnewenvironment{SQL}[1][] {
	\lstset{
		style=SQLstyle,
		commentstyle=\color{Brown},
		xleftmargin=0pt,
		framesep=5pt,
		#1
}}{}

\lstdefinestyle{SQLstyleReply}{%
	language=SQL,%
	basicstyle=\ttfamily\color{darkblue}\fontsize{8}{9}\selectfont,%
	deletekeywords={ON, BY},
	otherkeywords={with, WITH},%
	morekeywords={with, WITH},%
	morekeywords=[2]{mmu, inv, cpd},
	morekeywords=[3]{on},
	morekeywords=[4]{by},
	identifierstyle=\ttfamily\color{darkblue},%
	stringstyle=\ttfamily\color{darkblue},%
	showstringspaces=false,%
	mathescape=true,%
	upquote=true,%
}

\lstnewenvironment{SQLReply}[1][] {
	\lstset{
		style=SQLstyleReply,
		commentstyle=\color{Brown},
		xleftmargin=0pt,
		framesep=5pt,
		#1
}}{}

\def\schemacast{\rotatebox[origin=c]{0}{$\Delta$}}
\def\columncast{\rotatebox[origin=c]{0}{$\triangledown$}}

\SetAlCapNameFnt{\small}
\SetAlCapFnt{\small}

\copyrightyear{2020}
\acmYear{2020}
\setcopyright{acmcopyright}
\acmConference[SIGMOD'20]{Proc of the 2020 ACM SIGMOD
	International Conference on Management of Data}{June 14-19,
	2020}{Portland, OR, USA}
\acmPrice{15.00}
\acmDOI{10.1145/3318464.3389747}
\acmISBN{978-1-4503-6735-6/20/06}

\begin{document}
	\fancyhead{}

	\title{A Relational Matrix Algebra and its Implementation in a Column
		Store}
	
	\author{Oksana Dolmatova}
	\affiliation{
		\institution{University of Zürich}
		\city{Zürich}
		\country{Switzerland}
	}
	\email{dolmatova@ifi.uzh.ch}
	
	\author{Nikolaus Augsten}
	\affiliation{
		\institution{University of Salzburg}
		\city{Salzburg}
		\country{Austria}
	}
	\email{nikolaus.augsten@sbg.ac.at}
	
	\author{Michael H.\ Böhlen}
	\affiliation{
		\institution{University of Zürich}
		\city{Zürich}
		\country{Switzerland}
	}
	\email{boehlen@ifi.uzh.ch}
	
	
	\newmdenv[
	linewidth=2pt,
	roundcorner=5pt,
	leftmargin = 40,
	rightmargin = 40,
	outerlinecolor = blue!70!black,
	skipabove=10pt,
	skipbelow=2pt,
	]{reply}
	
	\definecolor{lightBlue}{RGB}{13,104,169}
	\definecolor{darkblue}{RGB}{0,102,204}
	\definecolor{dkgreen}{rgb}{0,0.6,0}
	
	\newcommand{\rev}[1]{{\color{black}#1}}
	
	\newcommand{\finrev}[1]{{\color{black}#1}}
	
	
	\begin{abstract}
		\rev{Analytical queries often require a mixture of
			relational and linear algebra operations applied to the same
			data. This poses a challenge to analytic systems that must bridge
			the gap between relations and matrices.  Previous work has mainly
			strived to fix the problem at the implementation level. This paper
			proposes a principled solution at the logical level. We introduce
			the \emph{relational matrix algebra} (RMA), which seamlessly
			integrates linear algebra operations into the relational model and
			eliminates the dichotomy between matrices and relations.
			RMA is closed: All our relational matrix operations are performed
			on relations and result in relations; no additional data structure
			is required.}
		Our implementation in MonetDB shows the feasibility of our approach,
		and empirical evaluations suggest that in-database analytics
		performs well for mixed workloads.
	\end{abstract}
	
	\maketitle
	
	\section{Introduction}
	\label{sec:introduction}
	
	Many data that are stored in relational databases include numerical
	parts that must be analyzed, for example, sensor data from industrial
	plants, scientific observations, or point of sales data.  \rev{The
		analysis of these data, which are not purely numerical but also
		include important non-numerical values, demand mixed queries that
		apply relational and linear algebra operations on the same data.}
	
	\rev{Dealing with mixed workloads is challenging since the gap between
		relations and matrices must be bridged.  Current relational systems
		are poorly equipped for this task.  Previous attempts to deal with
		mixed workloads have focused on the implementation level, for
		example, by introducing ordered data types; by storing matrices in
		special relations or key-value structures; or by splitting queries
		into their relational and matrix parts.  \finrev{This paper resolves
			the gap between relations and matrices.}
		
		We propose a principled solution for mixed workloads and introduce
		the \emph{relational matrix algebra} (RMA) to support complex data
		analysis within the relational model. The goal is to (1) solve the
		integration of relations and linear algebra at the logical level,
		(2) so achieve independence from the implementation at the physical
		level, and (3) prove the feasibility of our model by extending an
		existing system.  We are the first to achieve these goals: Other
		works focus on facilitating the transition between the relational
		and the linear algebra model. We eliminate the dichotomy between
		matrices and relations by seamlessly integrating linear algebra into
		the relational model. Our implementation of RMA in MonetDB shows the
		feasibility of our approach.
		
		We define linear operations over relations and systematically
		process and maintain non-numerical information.  We show that the
		relational model is well-suited for complex data analysis if
		ordering and contextual information are dealt with properly.}  RMA
	is purely based on relations and does not introduce any ordered data
	structures. Instead, the relevant row order for matrix operations is
	computed from \emph{contextual information} in the argument relations.
	All relational matrix operations return relations with
	\emph{origins}. Origins are constructed from the contextual
	information (attribute names and non-numerical values) of the input
	relations and uniquely identify and describe each cell in the result
	relation.
	
	\rev{We extend the syntax of SQL to support relational matrix
		operations.  As an example, consider a relation $rating$ with schema
		$(User, Balto, Heat, Net)$, that stores users and their ratings for
		the three films ("Balto", "Heat", and "Net", one column per
		film). The SQL query}
	\begin{SQL}
  SELECT * FROM inv(rating BY User);
	\end{SQL}
	\rev{orders the relation by users and computes the inversion of the
		matrix formed by the values of the ordered numerical columns. The
		result is a relation with the same schema: The values of attribute
		$User$ are preserved, and the values of the remaining three
		attributes are provided by matrix inversion (see
		Section~\ref{sec:RMA In Action} for details). The origin of a
		numerical result value is given by the user name in its table row
		and the attribute name of its column.}
	
	At the system level, we have integrated our solution into
	\mbox{MonetDB}.  Specifically, we extended the kernel with relational
	matrix operations implemented over binary association tables
	(BATs). The physical implementation of matrix operations is flexible
	and may be transparently delegated to specialized libraries that
	leverage the underlying hardware (e.g., MKL~\cite{mkl} for CPUs or
	cuBLAS~\cite{cuBLAS} for GPUs).  The new functionality is introduced
	without changing the main data structures and the processing pipeline
	of \mbox{MonetDB}, and without affecting existing functionality.
	
	Our technical contributions are as follows:
	\begin{itemize}
		\item We propose the \emph{relational matrix algebra} (RMA), which
		extends the relational model with matrix operations. This is the
		first approach \rev{to show that the relational model is sufficient
			to support matrix operations. The new set of operations is
			\emph{closed}: All relational matrix operations are performed on
			relations and result in relations, and no additional data
			structure is required.}
		\item We show that matrix operations are \emph{shape restricted},
		which allows us to \rev{systematically define the results of matrix
			operations over relations}. We define \emph{row and column
			origins}, the part of contextual information that \rev{describes
			values in the result relation}, and prove that all our operations
		return relations with origins.
		\item We implement and evaluate our solution in detail.  We show that
		our solution is feasible and leverages existing data structures and
		optimizations.
	\end{itemize}
	
	RMA opens new opportunities for advanced data analytics that combine
	relational and linear algebra functionality, speeds up analytical
	queries, and triggers the development of new logical and physical
	optimization techniques.
	
	The paper is organized as follows. Sect.~\ref{sec:Related Work}
	discusses related work.  We introduce basics in
	Sect.~\ref{sec:Background} and introduce relational matrix algebra
	(RMA) in Sect.~\ref{sec:reduct-line-oper}. \rev{We show an application
		example in Sect.~\ref{sec:RMA In Action}}, discuss important
	properties of RMA in Sect.~\ref{sec:Properties_RMA}, and its
	implementation in MonetDB in Sect.~\ref{sec:Implementation}. We
	evaluate our solution in Sect.~\ref{sec:Performance Evaluation} and
	conclude in Sect.~\ref{sec:Summary}.

	\section{Related Work}
	\label{sec:Related Work}
	
	Relational DBMSs offer simple linear algebra operations, such as the
	pair-wise addition of attribute values in a relation.  Some
	operations, e.g., matrix\footnote{Some approaches support
		multi-dimensional arrays. Since we target linear algebra, we focus
		on two dimensions and use the term \emph{matrix} throughout.}
	multiplication, can be expressed via syntactically complex and slow
	SQL queries.  The set of operations is limited and does not include
	operations whose results depend on the row order.  For instance, there
	are no SQL solutions for inversion or determinant computation.
	Complex operations must be programmed as UDFs.  Ordonez et al.\
	\cite{Ordonez} suggest UDFs for linear regression with a matrix-like
	result type.  The UTL\_NLA package~\cite{oracle_utl_nla} for
	\mbox{Oracle} DBMS offers linear algebra operations defined over
	UTL\_NLA\_ARRAY. UDFs provide a technical interface but do not define
	matrix operations over relations. No systematic approach to maintain
	contextual information is provided.
	
	Luo et al.~\cite{SimSQL_ext} extend \mbox{SimSQL} \cite{SimSQL}, a
	Hadoop-based relational system, with linear algebra functionality.
	RasDaMan \cite{intro_rasdaman, RasdamanAA} manages and processes image
	raster data.  Both systems introduce \emph{matrices} as ordered
	numeric-only attribute types.  Although relations and matrices
	coexist, operations are defined over different objects. Linear
	operations are not defined over unordered objects and they do not
	support contextual information for individual cells of a matrix.
	
	SciQL \cite{SciQL_introduction,SciQL_array_data_proc} extends MonetDB
	\cite{MonetDB} with a new data type, \mbox{\tt ARRAY}, as a
	first-class object.  An array is stored as an object on its own.
	Arrays have a fixed schema: The last attribute stores the data values
	of a matrix, all other attributes are dimension attributes and must be
	numeric. Arrays come with a limited set of operations, such as
	addition, filtering, and aggregation, and they must be converted to
	relations to perform relational operations.  The presence of
	contextual information and its inheritance are not addressed.
	
	The MADlib library \cite{MADlib} for in-database analytics offers a
	broad range of linear and statistical operations, defined as either
	UDFs with C++ implementations or Eigen library calls.  Matrix
	operations require a specific input format: Tables must have one
	attribute with a row id value and another array-valued attribute for
	matrix rows.  Matrix operations return purely numeric results and
	cannot be nested.
	
	Hutchison et al.\ \cite{LARADB} propose LARA, an algebra with
	tuple-wise operations, attribute-wise operations, and tuple
	extensions. LARA defines linear and relational algebra operations
	using the same set of primitives. This is a good basis for
	inter-algebra optimizations that span linear and relational
	operations.  LARA offers a strong theoretical basis, works out
	properties of the solution, and allows to store row and column
	descriptions during the operations. The maintenance of contextual
	information is not considered for operations that change the number of
	rows or columns.
	
	LevelHeaded \cite{LevelHeaded, EmptyHeaded}, an engine for relational
	and linear algebra operations, uses a special key-value structure:
	Each object has keys (dimension attributes) and annotations (value
	attributes).  Dimension and value attributes are stored in a trie and
	a flat columnar buffer, respectively.  Linear operations are available
	through an extended SQL syntax.  Key values guarantee contextual
	information for rows. However, the trie key structure restricts
	relational operations: For example, aggregations of keys and join
	predicates over non-key attributes (i.e., subselects in SQL) are not
	allowed.
	
	SciDB \cite{architecture_SciDB} is a DBMS that is based on
	arrays. Matrices and relations are implemented as nested arrays.
	SciDB focuses on efficient array processing and performs linear
	algebra operations over arrays. SciDB supports element-wise operations
	and selected linear operations, such as SVD. The system also offers
	relational algebra operations on arrays but cannot compete with
	relational DBMSs such as MonetDB in terms of performance.  A
	systematic approach to maintain contextual information is not
	considered.
	
	Statistical packages, such as R \cite{r_project} and pandas
	\cite{Pandas}, offer a broad range of linear and relational algebra
	operations over arrays.  Each cell may be associated with descriptive
	information, but this information is not always inherited as part of
	operations (e.g., $\usd$).  No systematic solution for associating
	contextual information with numeric results is provided. The most
	important relational operations are supported, but even basic
	optimizations (e.g., join ordering) are missing.
	
	The R package RIOT-DB~\cite{RIOTDB} uses MySQL as a backend and
	translates linear computations to SQL.  RIOT-DB addresses the main
	memory limitations of R, and the optimization of SQL statements yields
	inter-operation optimization.  However, it is difficult (or sometimes
	impossible) to express linear algebra operations in SQL, and only a
	few simple operations, such as subtraction and multiplication, are
	discussed.
	
	AIDA \cite{AIDA} integrates MonetDB and NumPy \cite{NumPy} and
	exploits the fact that both systems use C arrays as an internal data
	structure: To avoid copying NumPy data to MonetDB, AIDA passes
	pointers to arrays. Data copying is still needed to pass MonetDB
	results to NumPy since MonetDB does not guarantee that multiple
	columns are contiguous in memory, which is required by NumPy.  AIDA
	offers a Python-like procedural language for relational and linear
	operations.  Sequences of relational operations are evaluated lazily,
	which allows AIDA to combine and optimize sequences of relational
	operations.  The optimization does not include linear algebra
	operations.
	
	SystemML \cite{SystemMLMP} offers a set of linear algebra primitives
	that are expressed in a high-level, declarative language and are
	implemented on MapReduce.  SystemML includes linear algebra
	optimizations that are similar to relational optimizations (e.g.,
	selecting the order of execution of matrix multiplications).  The
	system considers only linear algebra operations.

	\section{Preliminaries}
	\label{sec:Background}
	
	This section presents notation for relations and matrices, and
	introduces the basic matrix algebra operations.
	
	\finrev{\subsection{Relations}}
	
	A relation $r$ is a set of tuples $r_i$ with schema $\mathcal{R}$.  A
	schema, $\mathcal{R} = (A, B, \ldots)$, is a finite, ordered set of
	attribute names.  A tuple $r_i\in r$ has a value from the appropriate
	domain for each attribute in the schema. We write $r_i.A$ to denote
	the value of attribute $A$ in tuple $r_i$ and $r.A$ to denote the set
	of all values $r_i.A$ in relation $r$. Ordered subsets of a schema,
	$\mathbf{U}\subseteq \mathcal{R}$, are typeset in bold.  $|r|$ is the
	number of tuples in relation $r$.
	
	Let $r$ be a relation and $\mathbf{U} \subseteq \mathcal{R}$ be
	attributes that form a key of $\mathcal{R}$.  We write
	$r^{\mathbf{U},k}$ to denote the \emph{$k\textsuperscript{th}$ tuple}
	of relation $r$ sorted by the values of attributes $\mathbf{U}$ (in
	ascending order):
	\begin{align} \label{eq:8}
	\begin{split}
	r_i = r^{\mathbf{U},k}  \iff\  
	& r_i \in r\ \land \\
	& |\{r_j \mid r_j \in r \land r_j.\mathbf{U} < r_i.\mathbf{U}\}| = k-1
	\end{split}
	\end{align}
	
	The \emph{column cast} $\columncast U$ creates an ordered
	set $L$ from the sorted values of an attribute $U$ that
	forms a key in relation $r$:
	\begin{align}\label{eq:9}
	\begin{split}
	L =  \columncast U  \iff \  
	& |L| = |r|\ \land \\
	& \forall 1 \leq i \leq |r| (L[i] = r^{U,i}.U)
	\end{split}
	\end{align}
	The column cast is used to generate a schema from a set of values.  We
	use this for operations $\tra$, $\usd$, and $\opd$ (see Table
	\ref{tbl:RMADef}).  The column cast is applicable if the cardinality
	of a list of attributes $\mathbf{U}$ is one.
	
	\begin{example}
		Consider relation $r$ in
		Figure~\ref{fig:relation-matrix-preliminaries}.  The third tuple of
		relation $r$ sorted by the values of attribute $V$ is
		$r^{(V),3} = (A,30,1)$, the column cast of $O$ is
		$\columncast O = (A,B,C)$, and the values of attribute
		$W$ are $r.W = \{1,5,1\}$.  \vspace*{-0pt}
		
		\begin{figure}[t] \centering {\fontsize{8}{9}\selectfont
				\begin{tabular} {|c|c|c|}
					\multicolumn{3}{l}{$r$} \\ \hline
					{\bf O} & {\bf V} & {\bf W} \\ \hline
					A  & 30 & 1 \\	\hline
					C  & 22 & 5 \\ \hline
					B  & 10  & 1 \\ \hline
				\end{tabular}
				\quad\quad
				\begin{tabular} {|m|c|}
					\multicolumn{2}{l}{$d$} \\ \hline
					\rowcolor{gray}
					{\bf} & {\bf 1} \\ \hline
					1 & D \\ \hline
					2 & B \\ \hline
				\end{tabular}
				\quad
				\begin{tabular} {|m|c|c|}
					\multicolumn{3}{l}{$e$} \\ \hline
					\rowcolor{gray}
					{\bf} & {\bf 1} & {\bf 2} \\ \hline
					1     & 1       & 3       \\	\hline
					2     & 2       & 4       \\ \hline
				\end{tabular}
				\quad
				\begin{tabular} {|m|c|c|c|}
					\multicolumn{4}{l}{ $d \mconc e$} \\ \hline
					\rowcolor{gray}
					{\bf} & {\bf 1} & {\bf 2} & {\bf 3} \\ \hline
					1     & D       & 1       & 3       \\	\hline
					2     & B       & 2       & 4       \\ \hline
			\end{tabular}}
			\caption{Relation $r$; matrices $d, e$, and $d \mconc e$}
			\label{fig:relation-matrix-preliminaries}
		\end{figure}
	\end{example}
	
	\vspace*{-0pt} We use set notation and apply it to bags.  Bags can be
	ordered or unordered.  To emphasize the difference, parentheses are
	used for ordered bags (or lists), e.g., (3,2,3), and curly braces 
	for unordered bags, e.g., \{3,2,3\}.  When
	transitioning from unordered to ordered bags, the order is specified
	explicitly.
	
	\finrev{\subsection{Matrices}}
	
	An $n \times k$ matrix $m$ is a two-dimensional array with $n$ rows
	and $k$ columns.  $|m|$ is the number of rows, $\#m$ the number of
	columns.  The element in the $i\textsuperscript{th}$ row and the
	$j\textsuperscript{th}$ column of matrix $m$ is $m[i,j]$; the
	$i\textsuperscript{th}$ row is $m[i,*]$; the $j\textsuperscript{th}$
	column is $m[*,j]$.
	
	We consider the operations from the R Matrix Algebra \cite{R_MA}:
	element-wise multiplication ($\emuMA$), matrix multiplication
	($\mmuMA$), outer product ($\opdMA$), cross product ($\cpdMA$), matrix
	addition ($\addMA$), matrix subtraction ($\subMA$), transpose
	($\traMA$), solve equation ($\solMA$), inversion ($\invMA$),
	eigenvectors ($\evcMA$), eigenvalues ($\evlMA$), QR decomposition
	($\qqrMA$, $\rqrMA$), SVD -- single value decomposition ($\dsdMA$,
	$\usdMA$, $\vsdMA$), determinant ($\detMA$), rank ($\rnkMA$), and
	Choleski factorization ($\chfMA$).  Note that QR and SVD return more
	than one matrix, therefore we split the operations: $\qqrMA$ and
	$\rqrMA$ return matrix Q and matrix R of the QR decomposition,
	respectively; $\dsdMA$, $\usdMA$, and $\vsdMA$ return vector $D$ with
	the singular values, matrix $U$ with the left singular vectors, and
	matrix $V$ with the right singular vectors of SVD, respectively.
	
	The \emph{matrix concatenation} of matrices $m$ and $n$ with $k$ rows
	each returns a matrix $h$ with $k$ rows.  The $i\textsuperscript{th}$
	row of $h$ is the concatenation of the $i\textsuperscript{th}$ row of
	$m$ and the $i\textsuperscript{th}$ row of $n$.
	\begin{align}\label{eq:10}
	\begin{split}
	h = m \mconc n 
	& \iff |h| = |m|\ \land \\
	& \forall 1 \leq i \leq |h| ( h[i,*] = m[i, *] \circ n[i, *])
	\end{split}
	\end{align}
	
	The \emph{schema cast} $\schemacast \mathbf{U}$ of attributes
	$\mathbf{U}$ creates a matrix $m$ (with a single column) from the
	attribute names of $\mathbf{U}$:
	\begin{align}
	\begin{split}
	m = \schemacast \mathbf{U}
	\iff\ 
	& \#m=1 \land  |m| = |\mathbf{U}|\ \land \\
	& \forall 1 \leq i \leq |\mathbf{U}| ( m[i,1] = \mathbf{U}[i] )
	\end{split}
	\end{align}
	
	\begin{example}
		Consider attributes $\mathbf{U} = (D, B)$.  Matrix $d$ in
		Figure~\ref{fig:relation-matrix-preliminaries} is the result of the
		schema cast $d = \schemacast\mathbf{U}$.  The result of
		concatenating matrix $d$ and matrix $e$ is $d \mconc e$. Note that
		the row and column numbers (cells shaded in gray) in the matrix
		illustrations are not part of the matrix.
	\end{example}
	
	Matrix operations are \emph{shape restricted}, i.e., the number of
	result rows is equal to the number of \emph{rows} of one of the input
	matrices (r), the number of \emph{columns} of one of the input
	matrices (c), or \emph{one} (1). The same holds for the number of
	result columns. 
	
	The dimensionality of result matrices defines the \emph{shape type} of
	matrix operations.  We write r$_1$ if the result dimensionality is
	equal to the number of rows in the first matrix, r$_2$ if the result
	dimensionality is equal to the number of rows in the second matrix,
	and r$_*$ if the result dimensionality is equal to the number of rows
	in the first and second matrix (i.e., r$_1$ = r$_2$).  The same
	notation holds for the number of columns.  Table
	\ref{tbl:card-matrix-ops} summarizes the shape types of matrix
	operations.
	
	\begin{table}[htbp] \centering
		\caption{Shape types of matrix operations}
		\label{tbl:card-matrix-ops}
		{\fontsize{8}{9}\selectfont
			\begin{tabular}{|l|c|c|} \hline \textbf{Cardinalities} &
				\textbf{Shape type} & \textbf{Operations} \\ \hline
				$|i_1 \times j_1| \rightarrow |i_1 \times i_1|$ & (r$_1$,r$_1$) & $\usdMA$  \\
				$|i_1 \times j_1|,|i_2 \times j_1| \rightarrow |i_1 \times i_2|$ & (r$_1$,r$_2$) & $\opdMA$   \\
				$|i_1 \times i_1| \rightarrow |i_1 \times i_1|$ & (r$_1$,c$_1$) & $\invMA$, $\evcMA$, $\chfMA$  \\
				$|i_1 \times j_1| \rightarrow |i_1 \times j_1|$ & (r$_1$,c$_1$) & $\qqrMA$  \\
				$|i_1 \times j_1|, |j_1 \times j_2| \rightarrow |i_1 \times j_2|$ & (r$_1$,c$_2$) &	$\mmuMA$ \\
				$|i_1 \times i_1| \rightarrow |i_1 \times 1|$ & (r$_1$,1) & $\evlMA$  \\
				$|i_1 \times j_1| \rightarrow |i_1 \times 1|$ & (r$_1$,1) & $\vsdMA$  \\
				$|i_1 \times j_1| \rightarrow |j_1 \times i_1|$ & (c$_1$,r$_1$) & $\traMA$  \\
				$|i_1 \times j_1| \rightarrow |j_1 \times j_1|$ & (c$_1$,c$_1$) & $\rqrMA$ , $\dsdMA$  \\
				$|i_1 \times j_1|,|i_1 \times j_2| \rightarrow |j_1 \times j_2|$ & (c$_1$,c$_2$) & $\cpdMA$  \\
				$|i_1 \times j_1|,|i_1 \times 1| \rightarrow |j_1 \times 1|$ & (c$_1$,c$_2$) & $\solMA$  \\
				$|i_1 \times j_1|, |i_1 \times j_1| \rightarrow |i_1 \times j_1|$ & (r$_*$,c$_*$) & $\emuMA$, $\addMA$, $\subMA$ \\
				$|i_1 \times i_1| \rightarrow |1 \times 1|$ & (1,1) & $\detMA$  \\
				$|i_1 \times j_1| \rightarrow |1 \times 1|$ & (1,1) & $\rnkMA$
				\\ \hline
		\end{tabular}}
	\end{table}	
	
	\begin{example}
		Matrix multiplication has shape type (r$_1$,c$_2$), which states
		that the number of result rows is equal to the number of rows of the
		first argument matrix, and the number of columns is equal to the
		number of columns of the second argument matrix.  Matrix addition
		has shape type (r$_*$,c$_*$), which states that the number of result
		rows is equal to the number of rows of the first matrix and the
		number of rows of the second matrix.
	\end{example}

	\rev{ All operations of the matrix algebra are \emph{shape
			restricted}.  This follows directly from the definitions of the
		matrix operations \cite{matrix_comp}.  The first column of
		Table~\ref{tbl:card-matrix-ops} lists the relevant cardinalities
		from these definitions.  We use shape restriction to determine the
		inheritance of contextual information.  It has also been used in
		size propagation techniques \cite{Boehm} for the purpose of
		cost-based optimization of chains of matrix operations.  }
	
	\section{Relational Matrix Algebra}
	\label{sec:reduct-line-oper}
	
	To seamlessly integrate matrix operations into the relational model,
	we extend the relational algebra to the \emph{relational matrix
		algebra} (RMA).  For each of the matrix operations we define a
	corresponding relational matrix operation in RMA: $\emu$, $\mmu$,
	$\opd$, $\cpd$, $\add$, $\sub$, $\tra$, $\sol$, $\inv$, $\evc$,
	$\evl$, $\qqr$, $\rqr$, $\dsd$, $\usd$, $\vsd$, $\det$, $\rnk$,
	$\chf$.  We use upper case for matrix operations (e.g., $\traMA$) and
	lower case for RMA operations (e.g., $\tra$).  RMA includes both relational
	algebra and relational matrix operations.  The new
	operations behave like regular operations with relations as input and
	output.
	
	For each argument relation, $r$, of a relational matrix operation
	\rev{one} parameter must be specified: The \emph{order schema}
	$\mathbf{U}\subseteq \mathcal{R}$ imposes an order on the tuples for
	the purpose of the operation.  The attributes of the order schema must
	form a key\footnote{\rev{Attributes that neither
			belong to the order schema nor the application schema must be
			dropped explicitly with a projection (or added to the order schema,
			thus, forming a super key)}.}.  \rev{The attributes of relation $r$ that are
		not part of the order schema $\overline{\mathbf{U}}$, i.e.,
		$\overline{\mathbf{U}} = \mathcal{R} - \mathbf{U} $ form the
		\emph{application schema}.}  The application schema identifies the
	attributes with the data to which the matrix operation is applied.
	
	The order schema $\mathbf{U}\subseteq\mathcal{R}$ splits
	relation $r$ into \rev{four} non-overlapping areas: \rev{\emph{Order
			schema} $\mathbf{U}$; \emph{order part} $r.\mathbf{U}$;
		\emph{application schema} $\overline{\mathbf{U}}$; and
		\emph{application part} $r.\overline{\mathbf{U}}$.}  The parts of
	$r$ that do not include matrix values, i.e., the order
	and application schemas ($\mathbf{U}$ and $\overline{\mathbf{U}}$) and
	the order part ($r.\mathbf{U}$), form the \emph{contextual
		information} for application part $r.\overline{\mathbf{U}}$.
	Intuitively, the order schema and application schema provide
	context for columns while the order part provides context for rows.

	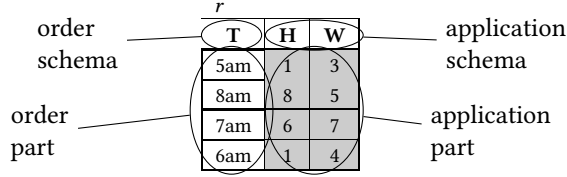
\begin{figure}[htbp]\centering
		\begin{tikzpicture}
		\node{ {\fontsize{8}{9}\selectfont
				\begin{tabular} {|c|a|a|}
				\multicolumn{3}{l}{$r$} \\ \hline
				\rowcolor{white} {\bf T}  & {\bf H} & {\bf W} \\ \hline
				5am  & 1 & 3 \\ \hline
				8am  & 8 & 5 \\ \hline
				7am  & 6 & 7 \\ \hline
				6am  & 1 & 4 \\ \hline
				\end{tabular}}
		};
		\node [shape = ellipse, draw, minimum height = .4cm, minimum width = 1.3cm] (eas) at (.45,.6) {};
		\node [shape = ellipse, draw, minimum height = .4cm, minimum width = 0.8cm] (eos) at (-.65,.6) {};
		\node [shape = ellipse, draw, minimum height = 1.7cm, minimum width = 1.1cm] (enap) at (-0.65,-0.4) {};
		\node [shape = ellipse, draw, minimum height = 1.7cm, minimum width = 1.3cm] (eap) at (0.45,-0.4) {};
		\node [align=left] (as) at (3cm,.5cm) {application\\ schema};
		\node [align=left] (ap) at (2.8cm,-.7cm) {application\\ part};
		\node [align=left] (nap) at (-3.2cm,-.7cm) {order \\ part};
		\node [align=left] (os) at (-2.7cm,.5cm) {order \\ schema};
		\draw (ap) -- (eap);
		\draw (nap) -- (enap);
		\draw (as) -- (eas);
		\draw (os) -- (eos);
		\end{tikzpicture}
		\caption{Structure of a relation instance}
		\label{fig:parts_of_a_rel}
	\end{figure}
	
	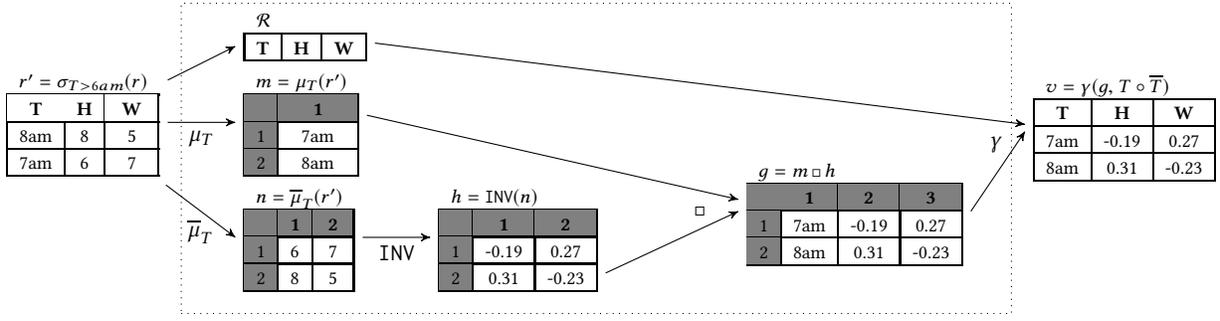
\begin{figure*}[hbtp] \centering \setlength{\tabcolsep}{5pt}
		\begin{tikzpicture}[scale=0.9,every node/.style={transform shape}]
		\node (IRR) [] { {\fontsize{8}{9}\selectfont
				\begin{tabular} {|c|c|c|}
				\multicolumn{3}{l}{\rev{$r' = \sigma_{T>6am}(r)$}}     \\ \hline
				\rowcolor{white} {\bf T} & {\bf H} & {\bf W} \\ \hline
				8am                          & 8        & 5        \\ \hline
				7am                          & 6        & 7        \\ \hline
				\end{tabular}}
		};
		
		\node (SCH) at ($(IRR.north east)+(1cm,.9cm)$) [anchor=north
		west] { {\fontsize{8}{9}\selectfont
				\begin{tabular} {|c|c|c|}
				\multicolumn{3}{l}{$\mathcal{R}$} \\ \hline
				{\bf T}  & {\bf H} & {\bf W} \\ \hline
				\end{tabular}}
		};
		
		\node (ODP) at ($(IRR.north east)+(1cm,0cm)$) [anchor=north west]
		{ {\fontsize{8}{9}\selectfont
				\begin{tabular} {|m|c|}
				\multicolumn{2}{l}{\rev{$m = \mu_T(r')$}} \\ \hline
				\rowcolor{gray}
				& {\bf 1}  \\ \hline
				1 & 7am    \\	\hline
				2 & 8am    \\	\hline
				\end{tabular}}
		};
		
		\node (OAP) at ($(IRR.north east)+(1cm,-1.7cm)$) [anchor=north
		west] { {\fontsize{8}{9}\selectfont
				\begin{tabular} {|m|c|c|}
				\multicolumn{3}{l}{\rev{$n = \overline{\mu}_T(r')$}} \\ \hline
				\rowcolor{gray} {\bf} & {\bf 1} & {\bf 2} \\ \hline
				1                     & 6       & 7       \\ \hline
				2                     & 8       & 5       \\ \hline
				\end{tabular}}
		};
		
		\node (RAP) [right = 1cm of OAP ] { {\fontsize{8}{9}\selectfont
				\begin{tabular} {|m|c|c|}
				\multicolumn{3}{l}{\rev{$h = \invMA(n)$}} \\ \hline
				\rowcolor{gray}	{\bf} & {\bf 1} & {\bf 2} \\ \hline
				1                     & -0.19     & 0.27     \\	\hline
				2                     & 0.31     & -0.23     \\ \hline
				\end{tabular}}
		};
		
		\node (concat) [below right = .2cm and 7.6cm of IRR] {$\mconc$};
		
		\node (CNC) [right = .3cm of concat ] { {\fontsize{8}{9}\selectfont
				\begin{tabular} {|m|c|c|c|}
				\multicolumn{3}{l}{\rev{$g = m \mconc h$}} \\ \hline
				\rowcolor{gray}	{\bf} & \rev{\textbf{1}} & \rev{{\bf 2}} & \rev{{\bf 3}} \\ \hline
				\rev{1}                     & \rev{7am} & \rev{-0.19}     & \rev{0.27}    \\	\hline
				\rev{2}                     & \rev{8am} & \rev{0.31}     & \rev{-0.23}    \\ \hline
				\end{tabular}}
		};
		
		\node (ORR) [above right = -.6cm and .8cm of CNC] {
			{\fontsize{8}{9}\selectfont
				\begin{tabular} {|c|c|c|}
				\multicolumn{3}{l}{\rev{$v = \gamma(g,T \circ \overline{T})$}} \\ \hline
				{\bf T}  & {\bf H} & {\bf W} \\ \hline
				7am      &  -0.19     & 0.27   \\	\hline
				8am      &   0.31     & -0.23  \\	\hline
				\end{tabular}}
		};
		
		\draw [->] (IRR) -- node[anchor=north] {} ($(SCH.west)+(0,-5pt)$);
		\draw [->] (OAP) -- node[anchor=north] {$\invMA$} (RAP);
		\draw [->] (IRR) -- node[anchor=north] {\rev{$\overline{\mu}_{T}$}} ($(OAP.west)+(0,0pt)$);
		\draw [->] (IRR) -- node[anchor=north] {\rev{$\mu_{T}$}} (ODP);
		\draw [->] ($(ODP.east)+(0,3pt)$) -- node[anchor=north] {} ($(CNC.west)+(0,5pt)$);
		\draw [->] (SCH) -- node[anchor=north] {} ($(ORR.west)+(0,1pt)$);
		\draw [->] ($(RAP.east)+(0,-15pt)$) -- node[anchor=north] {} ($(CNC.west)+(0,-0pt)$);
		\draw [->] ($(CNC.east)$) -- ($(ORR.west)+(0,-2pt)$);
		
		\node [below left = -.9cm and .2cm of ORR] {$\gamma$};
		
		\draw [dotted] ($(OAP.south west)+(-.8cm,-.2cm)$)
		rectangle	($(ORR.north west)+(-.2cm,.9cm)$);
		
		\end{tikzpicture}
		\caption{Structure of our solution for the inversion example,
			$v=\opr{\inv}{T}{}(\sigma_{T>6am}(r))$}
		\label{fig:SStr}
	\end{figure*}
	
	\begin{example}
		Order schema $\mathbf{U}=(T)$ splits relation $r$ in
		Figure~\ref{fig:parts_of_a_rel} into four parts: Order schema
		$\mathbf{U}=(T)$, application schema $\overline{\mathbf{U}}=(H,W)$,
		order part $r.\mathbf{U}=r.(T)=\{ 5am,$ $8am,$ $7am,$ $6am \}$, and
		application part $r.\overline{\mathbf{U}}=r.(H, W)=\{ (1,3),$
		$(8,5),$ $(6,7),$ $(1,4) \}$.
	\end{example}
	
	\subsection{Matrix and Relation Constructors}
	\label{sec:Matrix constructor}
	
	Figure~\ref{fig:SStr} summarizes our approach for the $\inv$ operation
	and example relation $r' = \sigma_{T>6am}(r)$: (1) Two matrix
	constructors define matrices $m$ and $n$ that correspond to order and
	application part of $r$, respectively; (2) $\invMA$ inverses matrix
	$n$ resulting in matrix $h$; and (3) the relation constructor combines
	$m \mconc h$ and $\mathcal{R}$ into result relation $v$.
	
	\begin{definition}{(\emph{Matrix constructor})}
		\rev{Let $r$ be a relation, $\mathbf{U}$ be an order schema.
			The matrix constructor $\mu_{\mathbf{U}}(r)$ returns a matrix that
			includes the values of $r.\mathbf{U}$ sorted by $\mathbf{U}$:}
		\begin{align*}
		m = \mu_{\mathbf{U}}(r)  \iff 
		& |m| = |r| \ \land \\
		& \forall 1 \leq i \leq |r| (m[i,*] = r^{\mathbf{U},i}.\mathbf{U})
		\\[-1cm]
		\end{align*}
		\label{def:mc}
	\end{definition}
	
	\rev{We use the complement notation $\overline{\mu}_{\mathbf{U}}(r)$
		to denote the matrix that includes the values of
		$r.\overline{\mathbf{U}}$ sorted by $\mathbf{U}$.}
	
	\begin{example} \label{ex:mcex}%
		Consider Figure~\ref{fig:SStr} with relation instance
		$r' = \sigma_{T>6am}(r)$ and schema $\mathcal{R} = (T, H, W)$.  The
		matrix constructor $\overline{\mu}_{T}(r')$ returns
		matrix $n$.
	\end{example}
	
	\begin{definition}{(\emph{Relation constructor})}
		Let $m$ be a matrix with unique rows, and $\mathcal{R}$ be a
		relation schema with $\#m$ attributes.  The relation constructor
		$\gamma(m, \mathcal{R})$ returns relation $r$ with schema
		$\mathcal{R}$:
		\begin{align*}
		r = \gamma(m, \mathcal{R}) 
		\iff 
		& |m| = |r| \ \land \\
		& \forall t ( t \in r \iff \exists 1 \leq i \leq |m| (t = m[i,*]))
		\\[-1cm]
		\end{align*}
		\label{def:rc}
	\end{definition}
	
	\begin{example}
		In Figure~\ref{fig:SStr}, a relation constructor is applied to
		schema $\mathcal{R}$ and the concatenated matrices $m \mconc h$ to
		construct the result relation:
		$v = \gamma(m \mconc h, \mathcal{R})$.
	\end{example}
	
	Matrix and relation constructors map between relations and
	matrices. We use constructors and matrices to define relational matrix
	operations and to analyze their properties. At the implementation
	level, constructors are very efficient since they split and combine
	lists of attribute names and \emph{do not} access the data (cf.\
	Section~\ref{sec:Implementation}).
	
	\subsection{Relational Matrix Operations}
	\label{sec: New basic operations}
	
	\begin{table*}[htbp] \centering
		\caption{Splitting and morphing relations and matrices}
		\label{tbl:RMADef}
		\setlength{\extrarowheight}{1pt}
		\def\arraystretch{1.07}
		\begin{tabular}{|c|c|l|} \hline
			\textbf{Shape type} & \textbf{Operations} & \textbf{Definition} \\ \hline
			(r$_1$,r$_1$)
			& $\usd$ &
			$\opr{\op}{\mathbf{U}}{}(r) =
			{\color{red}\gamma}(
			\opr{{\color{blue}\mu}}{\mathbf{U}}{}(r) \mconc
			\opMA(\opr{{\color{blue}\overline{\mu}}}{\mathbf{U}}{}(r)),
			\mathbf{U} \circ \columncast\mathbf{U})$
			\\
			\hline
			(r$_1$,r$_2$)
			& $\opd$ &
			$\opr{\op}{\mathbf{U};\mathbf{V}}{}(r,s) =
			{\color{red}\gamma}(
			\opr{{\color{blue}\mu}}{\mathbf{U}}{}(r)
			\mconc
			\opMA(\opr{{\color{blue}\overline{\mu}}}{\mathbf{U}}{}(r),
			\opr{{\color{blue}\overline{\mu}}}{\mathbf{V}}{}(s)),
			\mathbf{U} \circ \columncast\mathbf{V} )$
			\\
			\hline
			(r$_1$,c$_1$)
			& $\inv, \evc, \chf, \qqr$ &
			$\opr{\op}{\mathbf{U}}{}(r) =
			{\color{red}\gamma}(
			\opr{{\color{blue}\mu}}{\mathbf{U}}{}(r) \mconc
			\opMA(\opr{{\color{blue}\overline{\mu}}}{\mathbf{U}}{}(r)),
			\mathbf{U} \circ \overline{\mathbf{U}})$
			\\
			\hline
			(r$_1$,c$_2$)
			& $\mmu$ &
			$\opr{\op}{\mathbf{U};\mathbf{V}}{}(r,s) =
			{\color{red}\gamma}(
			\opr{{\color{blue}\mu}}{\mathbf{U}}{}(r)
			\mconc
			\opMA(\opr{{\color{blue}\overline{\mu}}}{\mathbf{U}}{}(r),
			\opr{{\color{blue}\overline{\mu}}}{\mathbf{V}}{}(s)),
			\mathbf{U} \circ \overline{\mathbf{V}} )$
			\\
			\hline		
			(r$_1$,1)
			& $\evl, \vsd$ &
			$\opr{\op}{\mathbf{U}}{}(r) =
			{\color{red}\gamma}(
			\opr{{\color{blue}\mu}}{\mathbf{U}}{}(r)
			\mconc
			\opMA(\opr{{\color{blue}\overline{\mu}}}{\mathbf{U}}{}(r)),
			\mathbf{U} \circ (\op))$
			\\
			\hline
			(c$_1$,r$_1$)
			& $\tra$ &
			$\oprc{\op}{\mathbf{U}}{}{C}(r) =
			{\color{red}\gamma}(
			\schemacast\overline{\mathbf{U}}
			\mconc
			\opMA(\opr{{\color{blue}\overline{\mu}}}{\mathbf{U}}{}(r)),
			(C) \circ \columncast\mathbf{U})$
			\\
			\hline
			(c$_1$,c$_1$)
			& $\rqr, \dsd$ &
			$\oprc{\op}{\mathbf{U}}{}{C}(r) =
			{\color{red}\gamma}(
			\schemacast\overline{\mathbf{U}}
			\mconc
			\opMA(\opr{{\color{blue}\overline{\mu}}}{\mathbf{U}}{}(r)),
			(C) \circ \overline{\mathbf{U}})$		
			\\
			\hline
			(c$_1$,c$_2$)
			& $\cpd, \sol$ &
			$\oprc{\op}{\mathbf{U};\mathbf{V}}{}{C}(r,s) =
			{\color{red}\gamma}(
			\schemacast\overline{\mathbf{U}}
			\mconc
			\opMA(\opr{{\color{blue}\overline{\mu}}}{\mathbf{U}}{}(r),
			\opr{{\color{blue}\overline{\mu}}}{\mathbf{V}}{}(s)),
			(C) \circ \overline{\mathbf{V}} )$
			\\
			\hline
			(r$_*$,c$_*$)
			& $\emu, \add, \sub$ &
			$\opr{\op}{\mathbf{U};\mathbf{V}}{}(r,s) =
			{\color{red}\gamma}(
			\opr{{\color{blue}\mu}}{\mathbf{U}}{}(r)
			\mconc
			\opr{{\color{blue}\mu}}{\mathbf{V}}{}(s)
			\mconc
			\opMA(\opr{{\color{blue}\overline{\mu}}}{\mathbf{U}}{}(r),
			\opr{{\color{blue}\overline{\mu}}}{\mathbf{V}}{}(s)),
			\mathbf{U} \circ \mathbf{V} \circ \overline{\mathbf{U}} )$
			\\
			\hline
			(1,1)
			& $\det, \rnk$ &
			$\oprc{\op}{\mathbf{U}}{}{C}(r) =
			{\color{red}\gamma}(
			r \circ \opMA(\opr{{\color{blue}\overline{\mu}}}{\mathbf{U}}{}(r)),
			(C,\op))$
			\\
			\hline
		\end{tabular}
	\end{table*}
	
	Relational matrix operations offer the functionality of matrix
	operations in a relational context.  The general form of a unary
	relational matrix operation is $\opr{\op}{\mathbf{U}}{}(r)$, where
	$\mathbf{U}$ is the order schema.  A binary operation
	$\opr{\op}{\mathbf{U}; \mathbf{V}}{}(r,s)$ has an additional order
	schema $\mathbf{V}$ for argument relation $s$.
	
	The result of a relational matrix operation is a relation that
	consists of (a) the \emph{base result} of the corresponding matrix
	operation, and (b) \emph{contextual information}, appropriately
	morphed from the contextual information of the argument relations
	to reflect the semantics of the operation.
	
	\begin{definition}{(\emph{Base result})} \label{def:base_res}%
		Consider a unary relational matrix operation
		$\opr{\op}{\mathbf{U}}{}(r)$.  The matrix that is the result
		of matrix operation $\opMA(\opr{\overline{\mu}}{\mathbf{U}}{}(r))$
		is the \emph{base result} of $\opr{\op}{\mathbf{U}}{}(r)$.
		The base result for binary operations is defined analogously.
	\end{definition}
	
	\begin{example}
		Consider $\opr{\inv}{T}{}(\sigma_{T>6am}(r))$ in
		Fig.~\ref{fig:SStr}.  The base result is matrix $h$, which results
		from $\invMA(\overline{\mu}_{T} (\sigma_{T>6am}(r)))$.
	\end{example}
	
	Table~\ref{tbl:RMADef} defines the details of how contextual
	information is maintained in relational matrix operations.  All
	definitions follow the structure illustrated in Figure
	\ref{fig:SStr}. A result relation is composed from order parts, base
	result, and schemas with the help of a relation constructor.  For
	example, $\inv$ is defined according to its shape tape in
	Table~\ref{tbl:RMADef}:
	$\opr{\inv}{\mathbf{U}}{}(r) = \gamma( \opr{\mu}{\mathbf{U}}{}(r)
	\mconc \invMA(\opr{\overline{\mu}}{\mathbf{U}}{}(r)), \mathbf{U} \circ
	\overline{\mathbf{U}})$, where $\opr{\mu}{\mathbf{U}}{}(r)$ are the
	rows of the order part,
	$\invMA(\opr{\overline{\mu}}{\mathbf{U}}{}(r))$ is the base result,
	and $ \mathbf{U} \circ \overline{\mathbf{U}}$ is the result schema.
	
	Operations that have a different number of rows than any of the input
	relations add a new attribute $C$ to the result relation. This
	attribute $C$ is for contextual information (cf.\
	Example~\ref{ex:inv}): Its values are either the attribute names of
	the application schema of an input relation or the operation name.
	The operations $\add, \sub, \emu$ require union compatible application
	schemas and non-overlapping order schemas.  Operations $\usd$,
	$\opd$, and $\tra$ construct the application schema of the result from
	the order schema of an input relation.  Therefore the
	cardinality of the order schemas $\mathbf{U}$ of $\tra$ and
	$\usd$, and $\mathbf{V}$ of $\opd$ must be one.
	
	\begin{example}
		\label{ex:inv}
		Consider Figures~\ref{fig:parts_of_a_rel} and \ref{fig:inv}.
		Figure~\ref{fig:qqr} illustrates the result of
		$\opr{\qqr}{T}{}(r)$.  The values of $T$ define the ordering of
		tuples for this operation.  The values of $H$ and $W$ are the values
		of matrix $Q$ computed as part of the QR decomposition.
		Figure~\ref{fig:transpose} illustrates the result of
		$\oprc{\tra}{T}{}{C}(r)$.  The column cast $\columncast T$ of
		ordering attribute $T$ provides names for the attributes in the
		transposed relation.  The result relation has a new attribute $C$
		whose values are the names of the attributes in the application
		schema of $r$.  Note that all result relations come with sufficient
		contextual information for each value.  For example, relation $r$ in
		Figure~\ref{fig:parts_of_a_rel} records that Humidity ($H$) was 1 at
		6am, which is also recorded in the transposed relation in
		Figure~\ref{fig:transpose}.
		
		\begin{figure}[htbp] \centering {\fontsize{8}{9}\selectfont
				\begin{subfigure}{0.5\linewidth}
					\setlength{\tabcolsep}{6pt}
					\centering
					\begin{tabular} {|c|c|c|}
						\multicolumn{3}{l}{\hspace*{-.1cm}$\opr{\qqr}{T}{}(r)$} \\ \hline
						{\bf T}  & {\bf H} & {\bf W} \\ \hline
						5am          & 0.1      & 0.5      \\	\hline
						6am          & 0.8      & -0.4     \\ \hline
						7am          & 0.6      & 0.4      \\	\hline
						8am          & 0.1      & 0.7      \\	\hline
					\end{tabular}
					\caption{QR decomposition}\label{fig:qqr}
				\end{subfigure}
				\hfill
				\begin{subfigure}{.45\linewidth} \setlength{\tabcolsep}{2pt}
					\begin{tabular} { | c | c | c | c | c |}
						\multicolumn{5}{l}{$\oprc{\tra}{T}{}{C}(r)$} \\ \hline
						\textbf{C} &  \textbf{5am} &  \textbf{6am} &  \textbf{7am}
						& \textbf{8am} \\ \hline
						H & 1 & 1 & 6 & 8 \\ \hline
						W & 3 & 4 & 7 & 5 \\ \hline
					\end{tabular}
					\caption{Transpose}\label{fig:transpose}
			\end{subfigure}}
			\caption{Examples of relational matrix operations}
			\label{fig:inv}
		\end{figure}
		
	\end{example}
	
	\section{RMA in Action}
	\label{sec:RMA In Action}
	
	\rev{This section gives an application example with a mixed workload
		that combines relational and linear algebra operations.  It
		maintains all data in regular relations and illustrates the
		importance of maintaining contextual information.}
	
	Consider relations $u$, $f$, and $r$ in Figure~\ref{fig:SMPLDATA}.
	Relation $u$ records name, state and year of birth of users;
	relation $f$ records title, release year and director of films;
	relation $r$ records user ratings for films.  Tuple $u_1$ states
	that user Ann lives in California and was born in 1980; tuple $f_1$
	states that film Heat was directed by Lee and was released in 1995;
	tuple $r_1$ states that Ann's ratings for Balto, Heat, and Net are,
	respectively, 2.0, 1.5, and 0.5. 
	
	\begin{figure}[htbp] \centering \setlength{\tabcolsep}{1.5pt}
		{\fontsize{7.5}{9}\selectfont
			\begin{tabular} {r|c|c|c|}
				\multicolumn{1}{l}{} & \multicolumn{3}{l}{\textit{u (user)}}  \\ \cline{2-4}	
				& \textbf{User} & \textbf{State} & \textbf{YoB} \\ \cline{2-4}
				$u_1$ & Ann & CA  & 1980 \\ \cline{2-4}	
				$u_2$ & Tom & FL  & 1965 \\ \cline{2-4}	
				$u_3$ & Jan & CA  & 1970 \\ \cline{2-4}		
			\end{tabular}
			\hspace{.0cm}
			\begin{tabular} {r|c|c|c|}
				\multicolumn{1}{l}{} & \multicolumn{3}{l}{\textit{f (film)}} \\ \cline{2-4}	
				& \textbf{Title}  & \textbf{RelY} & \textbf{Director} \\ \cline{2-4}
				$f_1$ & Heat  & 1995  & Lee \\ \cline{2-4}
				$f_2$ & Balto  & 1995  & Lee  \\ \cline{2-4}	
				$f_3$ & Net  & 1995  & Smith  \\ \cline{2-4}	
			\end{tabular}
			\hspace{.0cm}
			\begin{tabular} {r|c|c|c|c|}
				\multicolumn{1}{l}{} & \multicolumn{4}{l}{\textit{r (rating)}}\\ \cline{2-5}
				& \textbf{User}  & \textbf{Balto} & \textbf{Heat} & \textbf{Net} \\ \cline{2-5}
				$r_1$ & Ann &  2.0  & 1.5 & 0.5 \\ \cline{2-5}
				$r_2$ & Tom &  0.0  & 0.0 & 1.5\\ \cline{2-5}
				$r_3$ & Jan &  1.0 &  4.0 & 1.0 \\ \cline{2-5}
			\end{tabular}
		}
		\caption{Example database}
		\label{fig:SMPLDATA}
	\end{figure}
	
	The task is to determine how similar each of Lee's films is to any
	other film, based on the ratings from California users.  The
	covariance \cite{SamCov} is used to compute this similarity.  In
	addition, we need relational algebra operations (e.g., selection
	$\sigma$, aggregation $\vartheta$, rename $\rho$, and join $\Join$) to
	retrieve selected ratings and films, aggregate ratings, and combine
	information from different tables.  The key observation is that a
	mixture of matrix and relational operations is required to determine
	the similarities of the ratings.
	
	The solution in Figure~\ref{fig:SimilarityComputation} includes
	three key steps: Data preparation ($w1$), covariance computation
	($w2$-$w7$), and retrieving Lee's films together with all
	similarities ($w8$).  Note the seamless integration of linear and
	relational algebra\footnote{We use the first character of the
		attribute name to refer to attributes.}: The entire process
	frequently switches between linear and relational operations.  
	
	\begin{figure}[htbp] \centering
		\begin{align*}
		& w1 = \pi_{U, B, H, N}(\sigma_{S = \mq{CA}}(u \Join r)) \\
		& w2 = \vartheta_{AVG(B), AVG(H), AVG(N)}(w1) \\
		& w3 = \pi_{U, B, H, N}(\opr{\sub}{U;V}
		{}(w1, \rho_{V}(\pi_{U}(w1)) \times w2))  \\
		& w4 = \oprc{\tra}{U}{}{T}(w3) \\
		& w5 = \opr{\mmu}{C;U}{}(w4,w3) \\
		& w6 = w5 \times \rho_{M}(\vartheta_{COUNT(*)}(w1)) \\
		& w7 = \pi_{C,B/(M-1), H/(M-1), N/(M-1)}(w6) \\
		& w8 = \pi_{T, B, H, N}(\sigma_{D = \mq{Lee}} (w7 \Join_{C=T} f))\\[-20pt]
		\end{align*}
		\caption{Computing the similarity of the ratings}
		\label{fig:SimilarityComputation}
	\end{figure}
	
	In the following we discuss the algebra expressions in Figure
	\ref{fig:SimilarityComputation}.  First, we join $u$ and $r$ to
	select ratings from California users ($w1$).  Next, we compute the
	covariance using its standard definition \cite{SamCov}:
	$ cov(X,Y) =  \frac{1}{n-1}[(X-E[X])*(Y-E[Y])^{\mathsf{T}}]$.
	The expectation of an attribute, e.g.,
	$E(H) = \vartheta_{AVG(H)}(...)$, is computed via aggregation
	($w2$).  \emph{Relational matrix operations}, $\sub$, $\tra$ and
	$\mmu$, are used to subtract ($X-E[X]$), transpose ($^\mathsf{T}$),
	and multiply ($*$) relations ($w3, w4, w5$).  Next, we compute the
	unbiased covarinace ($w6, w7$).  Finally, we join $w7$ and $f$ to
	select Lee's films.
	
	Figure \ref{fig:CovSteps} illustrates relations $w3$, $w4$, and
	$w8$.  Consider transpose $\oprc{\tra}{U}{}{T}(w3)$ with order
	schema $U$ and application schema $\overline{U} = (B,H,N)$.  The
	result of this operation is a relation $w4$ with schema
	($C$,$Ann$,$Jan$).  The values of attribute $C$ are the attribute
	names in the application schema of $w3$.  Note that each operation
	preserves schema and ordering information as the crucial parts of
	contextual information.  This makes it possible to interpret the
	tuples in result relation $w8$.  For example, tuple $z_1$ states
	that Lee's film Balto has the smallest covariance to film Net.
	\vspace*{-5pt}
	
	\begin{figure}[htb] \centering \setlength{\tabcolsep}{1.5pt}
		{\fontsize{8}{9}\selectfont
			\begin{tabular} {|c|c|c|c|}
				\multicolumn{4}{l}{\textit{w3}}\\ \hline
				\textbf{U}  & \textbf{B} & \textbf{H} & \textbf{N} \\\hline
				Ann &  -1.25  & 0.5   & 0.25 \\ \hline
				Jan &  1.25  & -0.5   & 0.25\\ \hline
			\end{tabular}
			\hspace{.5cm}
			\begin{tabular} { |c|c|c| }
				\multicolumn{3}{l}{\textit{w4}} \\ \hline
				\textbf{C} & \textbf{Ann} & \textbf{Jan}  \\ \hline
				B          & -1.25     & 1.25          \\ \hline
				H          & -0.5      & 0.5           \\ \hline
				N          & -0.25    & 0.25            \\ \hline
			\end{tabular}
			\hspace{.2cm}
			\begin{tabular} {r|c|c|c|c|}
				\multicolumn{1}{l}{} & \multicolumn{4}{l}{\textit{w8}}\\ \cline{2-5}
				& \textbf{T}  & \textbf{B} & \textbf{H} & \textbf{N} \\ \cline{2-5}
				$z_1$ & B & 1.56   & -0.62  & -2.5  \\ \cline{2-5}
				$z_2$ & H &  -0.62  & 0.25    & 1 \\ \cline{2-5}
			\end{tabular}
		}
		\caption{Steps during the computation}
		\label{fig:CovSteps}
	\end{figure}  
	
	\section{Properties of RMA}
	\label{sec:Properties_RMA}
	
	This section defines two crucial requirements for relational matrix
	operations.  \emph{Matrix consistency} guarantees that the result of a
	relational matrix operation can be reduced to the result of the
	corresponding matrix operation.  \emph{Origins} guarantee that each
	result relation includes sufficient inherited contextual information
	\rev{to relate argument and result relation}.  We prove that each
	relational matrix operations is matrix consistent and returns a
	relation with origins.

	\subsection{Matrix Consistency}
	\label{sec:Matrix Consistency}	
	
	Matrix consistency ensures that the result relation includes all cell
	values that are present in the base result and the order of rows in
	the base result can be derived from contextual information in the
	result relations.  First, we define \emph{reducibility} to transition
	from relations to matrices.
	
	\begin{definition}{(\emph{Reducible})}
		Let $r$ be a relation, $\mathbf{U}$ be an order schema.  Relation $r$ is
		\emph{reducible} to matrix $m$ iff $m$ can be constructed from the
		attribute values of $\overline{\mathbf{U}}$ in relation $r$ sorted by
		$\mathbf{U}$: \vspace{-1em}
		\begin{center}
			$r \rightarrow_{\mathbf{U}} m \quad \Longleftrightarrow
			\quad \overline{\mu}_{\mathbf{U}}(r) = m$
		\end{center}
		\label{def:red}
	\end{definition}
	
	\begin{example}
		Consider Fig.~\ref{fig:SStr} with relation $r' = \sigma_{T>6am}(r)$,
		matrix $n$, and order schema $T$.  From Example~\ref{ex:mcex} we
		have $\overline{\mu}_{T}(r') = n$.  Relation $r'$ is reducible to
		matrix $n$ since $n$ can be constructed from the values of $H$ and
		$W$ in the argument relation sorted by $T$, i.e.,
		$r' \rightarrow_{T} n$.
	\end{example}
	
	\begin{definition}{(\emph{Matrix consistency})} Consider a unary
		matrix operation $\opMA(m)$.  The corresponding relational matrix
		operation $\op$ is \emph{matrix consistent} iff for all relations
		$r$ that are reducible to matrix $m$, the result relation
		$\opr{\op}{\mathbf{U}}{}(r)$ is reducible to $\opMA(m)$:
		\vspace{-0.5em}
		\begin{align*}
		\forall r,m,\mathbf{U}( 
		r \rightarrow_{\mathbf{U}} m\  \Longrightarrow \ 
		\exists  \mathbf{U'}( \opr{\op}{\mathbf{U}}{}(r)
		\rightarrow_{\mathbf{U'}} \opMA(m)))
		\end{align*}
		A binary relational matrix operation is matrix consistent if its
		result is reducible to the result of the corresponding binary matrix
		operation.
		\label{def:mcon}
	\end{definition}	
	
	\begin{example}
		Consider Figures~\ref{fig:parts_of_a_rel} and \ref{fig:mcexample}
		with relation $r$, matrix $g$, matrix $\rqrMA(g)$ and relation
		$\oprc{\rqr}{T}{}{C}(r)$.
		\begin{itemize}
			\item $r \rightarrow_{T} g$: Relation $r$ is reducible to matrix
			$g$
			\item $\oprc{\rqr}{T}{}{C}(r) \rightarrow_{C} \rqrMA(g)$:
			relation $\oprc{\rqr}{T}{}{C}(r)$ is reducible to matrix
			$\rqrMA(g)$
		\end{itemize}
		\vspace{-10pt}
		\begin{figure}[htbp] \centering%
			{\fontsize{8}{9}\selectfont
				\begin{tabular} {|m|c|c|}
					\multicolumn{3}{l}{$g$} \\ \hline
					\rowcolor{gray}	{\bf} & {\bf 1} & {\bf 2} \\ \hline
					1                     & 1     & 3    \\	\hline
					2                     & 1     & 4    \\ \hline
					3                     & 6     & 7     \\ \hline
					4                     & 8     & 5     \\ \hline
				\end{tabular}
				\quad
				\begin{tabular} {|m|c|c|}
					\multicolumn{3}{l}{$\rqrMA(g)$} \\ \hline
					\rowcolor{gray}	{\bf} & {\bf 1} & {\bf 2} \\ \hline
					1                     & -10.1     & -8.8     \\	\hline
					2                     & 0.0    & -4.6     \\ \hline
				\end{tabular}
				\quad
				\begin{tabular} {|c|c|c|}
					\multicolumn{3}{l}{$\oprc{\rqr}{T}{}{C}(r)$} \\ \hline
					{\bf C} & {\bf H} & {\bf W} \\ \hline
					H       & -10.1   & -8.8     \\	\hline
					W       & 0.0     & -4.6     \\ \hline
				\end{tabular}
			}
			\caption{Example of matrix consistency}
			\label{fig:mcexample}
		\end{figure}
	\end{example}
	
	\subsection{Origins of Result Relations}
	\label{sec:Origins_res_rel}
	
	The result of a relational matrix operation is a relation that, in
	addition to the base result, includes a \emph{row origin} and a
	\emph{column origin}.  Origins (1) uniquely define the relative
	positioning of result values, (2) give a meaning to values with
	respect to the applied operation, and (3) establish a connection
	between \rev{argument relations of an operation and its result relation}.
	
	\begin{example}
		Consider inversion and result relation $v$ in Figure \ref{fig:SStr}.
		Values $7am$ and $8am$ show that (1) value -0.19 precedes value 0.31
		because $7am$ precedes $8am$; (2) -0.19 is the inversion value for
		humidity and for time $7am$; (3) value -0.19 in relation $v$ is
		connected to value 6 in the argument relation since they have the
		same origins ($7am$ and H).
	\end{example}
	
	Origins are either inherited order schemas, application schemas from argument
	relations, or constants.  The shape type of an operation determines the
	cardinality of inherited contextual information.  The indices in the shape type
	specify the input relation, from which an origin is inherited. For example, if
	the first element of the shape type is c$_1$, the row origin is the schema cast
	of the application schema of the first argument relation. Note that indices
	$_*$ and $_2$ are only possible for binary operations.
	
	\begin{definition}{(\emph{Origins})}
		\label{def:origins}
		Consider a unary, $v = \opr{\op}{\mathbf{U}}{}(r)$, or binary,
		$v = \opr{\op}{\mathbf{U};\mathbf{V}}{}(r,s)$, matrix consistent
		operation with shape type ($x$,$y$), base result $m$, and attribute
		list ${\mathbf{U}'}$ such that $v \rightarrow_{\mathbf{U}'} m$.
		Consider Table~\ref{tbl:origins}.  $v.{\mathbf{U}'}$ is a \emph{row
			origin} iff $v.{\mathbf{U}'}$ is equal to \textit{ro} for the
		given shape type $x$.  $\overline{\mathbf{U}'}$ is a \emph{column
			origin} iff $\overline{\mathbf{U}'}$ is equal to \textit{co} for
		the given shape type $y$.
		
		\begin{table}[htbp] \centering
			\caption{Definition of origins for shape type ($x$,$y$)}
			\label{tbl:origins}
			{\fontsize{8}{9}\selectfont
				\begin{tabular}{|c|c|} \hline 
					\textbf{x} & \textit{ro} \\ \hline \hline
					r$_1$ &  $r.\mathbf{U}$ \\ \hline
					r$_2$ &  $s.\mathbf{V}$\\ \hline
					c$_1$ &  $\schemacast \overline{\mathbf{U}}$ \\ \hline
					c$_2$ &  $\schemacast \overline{\mathbf{V}}$ \\ \hline
					r$_*$ &  $(r.\mathbf{U}, s.\mathbf{V})$ \\ \hline
					c$_*$ &  $(\schemacast \overline{\mathbf{U}}, \schemacast \overline{\mathbf{V}})$ \\ \hline
					1 &   $'r'$ \\ \hline 
				\end{tabular}
				\qquad 		
				\begin{tabular}{|c|c|} \hline 
					\textbf{y} & \textit{co} \\ \hline \hline
					r$_1$ &  $\columncast \mathbf{U}$ \\ \hline
					r$_2$ &  $\columncast \mathbf{V}$\\ \hline
					c$_1$ &  $\overline{\mathbf{U}}$ \\ \hline
					c$_2$ &  $\overline{\mathbf{V}}$ \\ \hline
					r$_*$ &  $\columncast \mathbf{U}$ \\ \hline
					c$_*$ &  $\overline{\mathbf{U}}$ \\ \hline
					1 &   $'op'$ \\ \hline 
			\end{tabular}}
		\end{table}	
	\end{definition}
	
	\begin{example}
		\label{ex:origins}
		Figure \ref{fig:origins_example} illustrates relation $r$ and
		the origins for operations
		\begin{itemize}\itemsep=3pt
			\item $\oprc{\rnk}{H}{}{R}(\pi_{H,W}(r))$ with shape type (1,1)
			\item $\opr{\usd}{T}{}(r)$ with shape type (r$_1$,r$_1$)
			\item $\opr{\qqr}{W,T}{}(r)$ with shape type (r$_1$,c$_1$)
		\end{itemize}
		\emph{Column origins} ($co$) are marked by rectangles (all values
		inside a rectangle form together the column origin for the
		relation).  \emph{Row origins} ($ro$) are marked by ellipses.
		
		\begin{figure}[!htb] \centering%
			\setlength{\tabcolsep}{2pt}
			{\fontsize{8}{9}\selectfont
				\begin{tikzpicture}
				
				\node (a) [] {
					\begin{tabular} {|c|c|c|}
					\multicolumn{3}{l}{$r$} \\ \hline
					\rowcolor{white} {\bf T}  & {\bf H} & {\bf W} \\ \hline
					5am  & 1 & 3 \\ \hline
					8am  & 8 & 5 \\ \hline
					7am  & 6 & 7 \\ \hline
					6am  & 1 & 4 \\ \hline
					\end{tabular}
				};
				
				\node (c) [above = .5cm of a] {
					\begin{tabular} {|c|c|}
					\multicolumn{2}{l}{$p1 = \oprc{\rnk}{H}{}{R}(\pi_{H,W}(r))$} \\ \hline
					{\bf C} & {\bf \rectangled{rnk}} \\ \hline
					\ellipsed{$\; r \; $} & 1 \\ \hline
					\multicolumn{2}{l}{$ro$  = $r$ = (r) }\\
					\multicolumn{2}{l}{$co$ = $\op$ = (rnk)}
					\end{tabular}
				};
				
				\node (b) [below right = -2.3cm and .5cm  of c] {
					\begin{tabular} {|c|c|c|c|c|}
					\multicolumn{5}{l}{$p2 = \opr{\usd}{T}{}(r)$} \\ \hline
					{\bf T} & \rectangled{\textbf{5am}} & \rectangled{\textbf{6am}}
					& \rectangled{\textbf{7am}} & \rectangled{\textbf{8am}} \\ \hline
					\ellipsed{5am}  & -0.2 & 0.5  & -0.8 & 0.4 \\	\hline
					\ellipsed{6am}  & -0.3 & 0.6 & 0.6 & 0.4 \\ \hline
					\ellipsed{7am}  & -0.7 & 0.2 & 0.0 & -0.7 \\	\hline
					\ellipsed{8am}  & -0.7 & -0.6 &0.0 & 0.4 \\	\hline
					\multicolumn{5}{l}{$ro$ = $r.\mathbf{U}$ = (5am, 6am, 7am, 8am)}\\
					\multicolumn{5}{l}{$co$  = $\columncast\mathbf{U}$ = (5am, 6am, 7am, 8am)}
					\end{tabular}
				};
				
				\node (d) [below = 0cm of b] {
					\begin{tabular} {|c|c|c|l}
					\multicolumn{3}{l}{$p3 = \opr{\qqr}{W,T}{}(r)$} \\ \cline{1-3}
					{\bf W} & {\bf T}  & \rectangled{{\bf H}} \\ \cline{1-3}
					3       & 5am      & 0.1      \\ \cline{1-3}
					4       & 6am      & 0.1     \\ \cline{1-3}
					5       & 8am      & 0.7      \\ \cline{1-3}
					7       & 7am      & 0.5      \\ \cline{1-3}
					
					\multicolumn{4}{l}{\hspace*{-2cm}$ro$ = $r.\mathbf{U}$  =
						((3,5am), (4,6am), (5,8am), (7,7am))}\\
					\multicolumn{4}{l}{\hspace*{-2cm}$co$  = $\overline{\mathbf{U}}$ = (H)}
					\end{tabular}
				};
				
				\draw [->] (a) -- node[anchor=north] {} (b);
				\draw [->] (a) -- node[anchor=north] {} (c);
				\draw [->] (a) -- node[anchor=north] {} (d);
				
				\node [shape = ellipse, draw, minimum height = .35cm, minimum width = 1.2cm] (es) at (3,-0.45) {};
				\node [shape = ellipse, draw, minimum height = .35cm, minimum width = 1.2cm] (es) at (3,-.85) {};
				\node [shape = ellipse, draw, minimum height = .35cm, minimum width = 1.2cm] (es) at (3,-1.22) {};
				\node [shape = ellipse, draw, minimum height = .35cm, minimum width = 1.2cm] (es) at (3,-1.6) {};
				\end{tikzpicture}
			}
			\vspace*{-10pt}
			\caption{Examples of origins}
			\label{fig:origins_example}
		\end{figure}
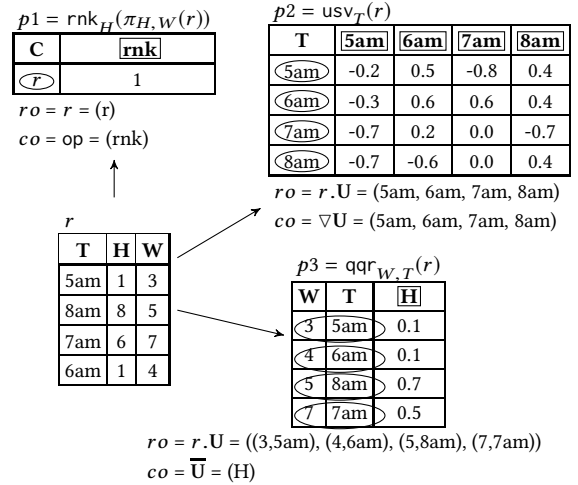
		
		For $p2 = \opr{\usd}{T}{}(r)$, we have $\mathbf{U} = (T)$,
		$\overline{\mathbf{U}} = (H,W)$, $\mathbf{U}' = (T)$, and
		$\overline{\mathbf{U'}}$ = (5am, 6am, 7am, 8am). The shape type of
		$\usd$ is (r$_1$, r$_1$) (see Table \ref{tbl:RMADef}) this makes
		$p2.T = r.T$ a row origin and $\columncast T$ = (5am, 6am, 7am, 8am) a
		column origin.
	\end{example}

	\subsection{Correctness}
	\label{sec:Correctness}	
	
	\begin{theorem}
		\label{lem:rel-con-inf}
		All relational matrix operations return matrix consistent relations
		with a row and column origin.
	\end{theorem}
	
	 \begin{proof}%
	   First, we prove that relational matrix operations are
	   matrix consistent.  Second, we show that inherited contextual 
	   information in the result  corresponds to
	   Definition~\ref{def:origins}.
	
	   (1) Consider a unary operation with shape type ($x$,$y$).  We start
	   with the definition of matrix consistency
	   (Definition~\ref{def:mcon}), $\mathbf{U}'$ equal to $\mathbf{U}$
	   ($x$ = r$_1$) or $C$ (x = c$_1$), and $\overline{\mathbf{U}'}$ equal to
	   $\overline{\mathbf{U}}$ (y = c$_1$) or $\columncast {\mathbf{U}}$ (y = r$_1$).
	   We instantiate the implication with the definition of the operation
	   in Table~\ref{tbl:RMADef}. Then, we simplify the right hand side of
	   the implication, substituting it with the equality in
	   Definition~\ref{def:red}. Next, we expand the equality with the
	   definitions of relational and matrix constructors
	   (Definitions~\ref{def:mc} and ~\ref{def:rc}).  A series of
	   simplifications yields the equality of the right and left hand side.
	
	   (2) Consider operation $v = \opr{\qqr}{\mathbf{U}}{}(r)$ with shape
	   type (r$_1$, c$_1$). Matrix consistency of $\qqr$ was shown in the
	   first part of the proof.  The shape type of $\qqr$ is (r$_1$,c$_1$)
	   and we get $\overline{\mathbf{U}} = \overline{\mathbf{U}'}$, and
	   $\mathbf{U} = \mathbf{U}'$.  $v.\mathbf{U}$ is the row origin
	   because x=r$_1$ and $\overline{\mathbf{U}}$ is the column origin
	   because y=c$_1$ (see Table \ref{tbl:origins}).
	
	   The same reasoning applies to the other types of operations in
	   Table~\ref{tbl:RMADef}.  Thus, each relational matrix operation
	   returns a result relation with a row and a column origin.
	 \end{proof}
	
	The following example uses a sequence of $\opr{\tra}{}{}$ operations
	to illustrate the importance of origins.  A result relation with
	origins inherits sufficient contextual information, such that each
	value can be interpreted.  Origins also carry sufficient information
	about the order of rows, such that in sequences of relational matrix
	operations no ordering information is lost between operations.
	
	\begin{example}
		Consider a relation instance $r$ that is reducible to matrix $n$.
		Figure~\ref{fig:IStr} shows the matrix expression
		$\traMA(\traMA(n))$ and the respective relational matrix
		expression
		$\oprc{\tra}{C}{}{C}(\oprc{\tra}{T}{}{C}(r))$.
		Operation $_{T}\tra_{}(r)$ returns relation $r1$, which in
		addition to the application schema (attributes $5am$, $6am$, $7am$,
		$8am$) also includes attribute $C$, which is preserved together with
		the application schema.
		
		\begin{figure}[ht]\centering\small
			\setlength{\tabcolsep}{2pt}
			\begin{tikzpicture}[]
			\node[align=center] (r1) { {\fontsize{8}{9}\selectfont
					\begin{tabular} {|c|a|a|}
					\multicolumn{3}{l}{$r$} \\ \hline
					\rowcolor{white} \textbf{T}  & \textbf{H}
					& \textbf{W} \\ \hline
					5am &  1 & 3 \\ \hline
					8am &  8 & 5 \\ \hline
					7am & 6 & 7 \\ \hline
					6am & 1 & 4 \\ \hline
					\end{tabular}}
			};
			\node (m1) at ($(r1.north east)+(3.0cm,0cm)$) [anchor=north west] {
				\setlength{\tabcolsep}{4pt}{\fontsize{8}{9}\selectfont
					\begin{tabular} {|m|c|c|}
					\multicolumn{3}{l}{$n$} \\ \hline
					\rowcolor{gray} {\bf} &{\bf 1} &  {\bf 2}  \\ \hline
					1   & 1 &    3 \\ \hline
					2   & 1 &    4 \\ \hline
					3   & 6 &    7 \\ \hline
					4   & 8 &    5 \\ \hline
					\end{tabular}}
			};
			\node[align=center] (r2) at ($(r1.south)+(-0.0cm,-.6cm)$) [anchor=north]{
				{\fontsize{8}{9}\selectfont
					\begin{tabular} {|c|a|a|a|a|}
					\multicolumn{3}{l}{$r1$} \\ \hline
					\rowcolor{white} {\bf C} & {\bf 5am} & {\bf 6am}
					& {\bf 7am} & {\bf 8am} \\ \hline
					H & 1 & 1 & 6 & 8 \\ \hline
					W & 3 & 4 & 7 & 5 \\ \hline
					\end{tabular}}
			};
			\node (m2) at ($(m1.south east)+(0.46cm,-.6cm)$) [anchor=north east] {
				\setlength{\tabcolsep}{4pt}{\fontsize{8}{9}\selectfont
					\begin{tabular} {|m|c|c|c|c|}
					\multicolumn{5}{l}{$n1$} \\ \hline
					\rowcolor{gray} {\bf} & {\bf 1} & {\bf 2} & {\bf 3} & {\bf 4}  \\ \hline
					1 & 1 & 1 & 6 & 8 \\ \hline
					2 & 3 & 4 & 7 & 5 \\ \hline
					\end{tabular}}
			};
			\node[align=center] (r3) at ($(r2.south)+(0.0cm,-.6cm)$) [anchor=north]{
				{\fontsize{8}{9}\selectfont
					\begin{tabular} {|c|a|a|}
					\multicolumn{3}{l}{$r2$} \\ \hline
					\rowcolor{white} {\bf C} & {\bf H} & {\bf W} \\ \hline
					5am       & 1       & 3         \\ \hline
					6am       & 1       & 4         \\ \hline
					7am       & 6       & 7         \\ \hline
					8am       & 8       & 5         \\ \hline
					\end{tabular}}
			};
			\node (m3) at ($(m2.south east)+(-0.46cm,-.6cm)$) [anchor=north east] {
				\setlength{\tabcolsep}{4pt}{\fontsize{8}{9}\selectfont
					\begin{tabular} {|m|c|c|}
					\multicolumn{3}{l}{$n2$} \\ \hline
					\rowcolor{gray} {\bf} &{\bf 1} &  {\bf 2}  \\ \hline
					1   & 1 &    3 \\ \hline
					2   & 1 &    4 \\ \hline
					3   & 6 &    7 \\ \hline
					4   & 8 &    5 \\ \hline
					\end{tabular}}
			};
			
			\draw[arrow] (r1.south)--(r2.north) node[pos=.5, right] {$\oprc{\tra}{T}{}{C}(r)$};
			\draw[arrow] (r2.south)--(r3.north) node[pos=.5, right] {$\oprc{\tra}{C}{}{C}(r1)$};
			
			\draw[arrow] (r1.east)--(m1.west) node[pos=.5, anchor = south, below]
			{$r \rightarrow_{T} n$};
			\draw[arrow] (r2.east)--(m2.west) node[pos=.5, anchor = south, below]
			{$r1 \rightarrow_{C} n1$};
			\draw[arrow] (r3.east)--(m3.west) node[pos=.5, anchor = south,below]
			{$r2 \rightarrow_{C} n2$};
			
			\draw[arrow] (m1.south)--(m2.north) node[pos=.5,left]{$\traMA(n)$};
			\draw[arrow] (m2.south)--(m3.north) node[pos=.5,left]{$\traMA(n1)$};
			
			\end{tikzpicture}
			\caption{Origins and matrix consistency}
			\label{fig:IStr}
		\end{figure}
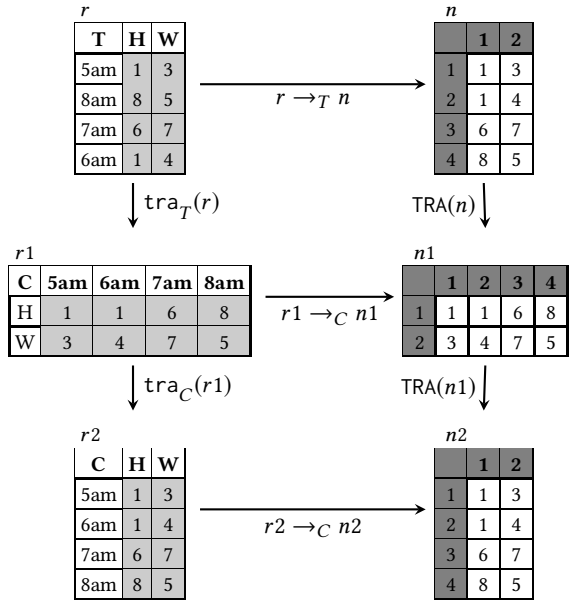
		
	\end{example}

	\section{Implementation}
	\label{sec:Implementation}
	
	We discuss the integration of our solution into
	\mbox{MonetDB}.  The implementation of relational matrix operations
	includes the processing of contextual information and the computation
	of the base result.  Contextual information is handled inside
	\mbox{MonetDB}, while the computation of the base result can be done
	in \mbox{MonetDB} or delegated to external libraries (e.g., MKL).  
	The integration of each relational matrix operation requires
	extensions throughout the system, but does not change the query
	processing pipeline \rev{and no new data structures are introduced. To
		extend \mbox{MonetDB} with addition, QR decomposition, linear
		regression, and the transformation of numerical data to the MKL
		format we touch 20 (out of 4500) files and add 2500 lines of code.}
	
	\subsection{ \mbox{MonetDB} }
	\label{sec:monetdb}
	
	\mbox{MonetDB} stores each column of a table as a binary association
	table (BAT). A BAT is a table with two columns: Head and tail. The
	head is a column with object identifiers (OID), while the tail is a
	column with attribute values.  All attribute values of a tuple in a
	relation have the same OID value.  Thus, a tuple can be constructed by
	concatenating all tail values with the same OID.  \mbox{MonetDB}
	operations manipulate BATs and relational operations are represented
	and executed as sequences of BAT operations. Example BAT operations
	are $B_1*B_2$, $B_1/B_2$, and $B_1-B_2$ for element-wise multiplication,
	division, and subtraction, and $sum(B)$ to sum the values in BAT $B$.
	
	 \begin{example}
	  \mbox{MonetDB} stores relation $\sigma_{T>6am}(r)$ from
	 Figure~\ref{fig:SStr} in three BATs as shown in
	 Figure~\ref{fig:STTree}.
	
	 	\begin{figure}[!htb] \centering \setlength{\tabcolsep}{2pt}
	 	{\fontsize{8}{9}\selectfont
	 		\begin{tabular} { |c|c| }
	 			\multicolumn{2}{l}{T} \\ \hline			
	 			{\bf OID} & {\bf Val} \\ \hline
	 			0         & 8am     \\	\hline
	 			1         & 7am     \\	\hline
	 		\end{tabular}}
	 	\quad
	 	{\fontsize{8}{9}\selectfont
	 		\begin{tabular} { |c|c| }
	 			\multicolumn{2}{l}{H} \\ \hline
	 			{\bf OID} & {\bf Val} \\ \hline
	 			0         & 8        \\	\hline
	 			1         & 6       \\	\hline
	 	\end{tabular}}
	 	\quad
	 	{\fontsize{8}{9}\selectfont
	 		\begin{tabular} { |c|c| }
	 			\multicolumn{2}{l}{W} \\ \hline	
	 			{\bf OID} & {\bf Val} \\ \hline
	 			0         & 5         \\	\hline
	 			1         & 7         \\	\hline
	 	\end{tabular}}
	  \caption{BAT representation of $\sigma_{T>6am}(r)$}
	 	\label{fig:BATREPR}
	 \end{figure}
	 \end{example}
	
	One important BAT operation is \emph{leftfetchjoin} ($\downarrow$),
	which returns a BAT with OIDs sorted according to the order of OIDs of
	another BAT from the same relation.  For instance, $X\!\downarrow\! Y$
	returns BAT $X$, whose OIDs have the same order as OIDs of BAT $Y$.
	$X\!\downarrow\!X$ denotes $X$ sorted by its own values.
	
	\subsection{RMA Integration}
	\label{sec:RMA Integration}
	
	As a first step, we have extended the SQL parser to make the
	relational matrix operations available in the from clause of SQL
	\finrev{\cite{RMA_nonanonymus}}.  The syntax \texttt{(r BY U)} the
	specifies ordering for an argument relation $r$.  As an example,
	consider relations $r$ and $s$ and ordering attributes $\mathbf{U}$
	and $\mathbf{V}$.  The unary operation $\opr{\inv}{\mathbf{U}}{}$ and
	the binary operation $\opr{\mmu}{\mathbf{U};\mathbf{V}}{}$ are
	expressed as:
	
	\begin{SQL}
SELECT * FROM inv(r BY $\mathbf{U}$);
SELECT * FROM mmu(r BY $\mathbf{U}$, s BY $\mathbf{V}$);
	\end{SQL}
	
	These basic constructs can be composed to more complex expressions.
	For instance, folding $w5$, $w6$ and $w7$ from
	Figure~\ref{fig:SimilarityComputation} yields the RMA expression:
	\begin{align*}
	\pi_{C,B/(M-1),  H/(M-1), N/(M-1)}&(\\
	\opr{\mmu}{C;U}{}(w4, w3)& \times
	\rho_{M}(\vartheta_{COUNT(*)}(w1)))
	\end{align*}
	
	\smallskip
	\noindent
	The SQL translation of this expression is:
	
	\smallskip
	\begin{SQL}
SELECT C, B/(M-1), H/(M-1), N/(M-1)
FROM mmu(w4 BY C, w3 BY U) AS w5
     cross join
     ( SELECT COUNT(*) AS M FROM w1 ) AS t;		  
	\end{SQL}
	\medskip
	
	Algorithm~\ref{alg:STTREE} processes a node that represents a unary
	relational matrix operation $\opr{\op}{\mathbf{U}}{}(r)$ and
	translates it to a list of BAT expressions.  In lines
	\ref{lst:ssmstart} - \ref{lst:ssmend}, the BATs of relation $r$ are
	split, sorted, and morphed to get BATs $X$ with row origins and BATs
	$Y$ with the application part. \emph{Splitting} (lines
	\ref{lst:ssmstart} and \ref{lst:sortend}) divides a relation into two
	parts: The application part, on which the matrix operations are
	performed, and the contextual information, which gives a meaning to
	the application part.  BATs $B$ are split into application part and
	order part according to $\mathbf{U}$.  \emph{Sorting} (lines
	\ref{lst:sortstart} and \ref{lst:sortend}) determines the order of the
	tuples for a specific matrix operation.  The order schema $\mathbf{U}$
	is used to sort the BATs: BATs in $\mathbf{U}$ are sorted according to
	their values while the other BATs in $B$ are sorted according to the
	OIDs of the BATs in $\mathbf{U}$.  The order is established for each
	operation based on the contextual values in the
	relation. \emph{Morphing} (lines
	\ref{lst:morphstart}-\ref{lst:ssmend}) morphs contextual information
	so that it can be added to the base result.  Finally, the matrix
	operation is applied to $Y$ (line \ref{lst:matrop}).  \emph{Merging}
	(line \ref{lst:return}) combines the result of a matrix operation with
	relevant contextual information and constructs the result relation
	with row and column origins.  Merging and splitting are efficient
	operations that work at the schema level and do not access the data.
	
	\begin{algorithm2e}[ht]
		\small
		\caption{UnaryRMA($\op$,  $\mathbf{U}$, $r$)}
		\label{alg:STTREE}
		$B \gets BATs(r);$
		$Y \gets \{\} $ \;
		$D \gets \{b | b \in B,  b \textit{ is in order schema } \mathbf{U}\} $ \; \label{lst:ssmstart}
		$G \gets sort(D)$ \; \label{lst:sortstart}
		\lFor{$b \in B \setminus D$}{$Y \gets Y \cup b\!\downarrow\!G$} \label{lst:sortend}
		\lIf{ShapeType($\op$) $\in$ \{(r,r), (r,c), (r,1)\}}{$X \gets G$} \label{lst:morphstart}
		\lElseIf{ShapeType($\op$) $\in$ \{(c,r), (c,c)\}}{$X \gets new BAT(Y)$}
		\lElse{$X \gets new BAT(r)$} \label{lst:ssmend}
		$F \gets eval(\op, Y)$\;  \label{lst:matrop}
		\Return $Concat(X,F)$\; \label{lst:return}
	\end{algorithm2e}
	
	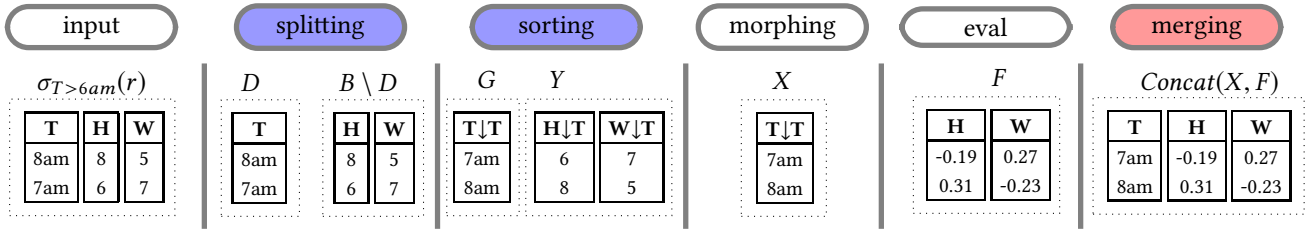
\begin{figure*}[ht]\centering
		\setlength{\tabcolsep}{3pt}
		\begin{tikzpicture}[head/.style={line width=2pt, draw=gray, rounded
			corners=8pt,text width=2cm,align=center}]
		
		\node[head,fill=white] (in) {input};
		\node[fill=white!40,below=0.15cm of in] (al1) {$\sigma_{T>6am}(r)$};
		\node[fill=white!40,below=0.7cm of in] (iv) [] 
		{\fontsize{8}{9}\selectfont
			\setlength{\tabcolsep}{3pt}
			\begin{tabular}{|c|} \hline
			{\bf T} \\ \hline
			8am \\ 
			7am \\ \hline
			\end{tabular}
			\begin{tabular} {|c|} \hline
			{\bf H}   \\ \hline
			8 \\ 
			6 \\ \hline
			\end{tabular}
			\begin{tabular} {|c|} \hline
			{\bf W}   \\ \hline
			5 \\ 
			7 \\ \hline
			\end{tabular}
		};
		\draw [dotted] ($(iv.south west)+(-.00cm,-.05cm)$)
		rectangle ($(iv.north east)+(.02cm,.05cm)$);	
		
		\node[head,fill=blue!40,right=0.74cm of in] (sp) {splitting};
		\node[fill=white!40,below left=0.15cm and -2.3cm of sp] (cd) {$ D \qquad \quad  B \setminus D$};
		\node[fill=white!40,below=0.7cm of sp] (c1d1) [] 
		{\fontsize{8}{9}\selectfont
			\setlength{\tabcolsep}{3pt}
			\begin{tabular}{|c|} \hline
			{\bf T} \\ \hline
			8am \\ 
			7am \\ \hline
			\end{tabular}
			\quad
			\quad
			\begin{tabular} {|c|}
			\hline
			{\bf H}   \\ \hline
			8 \\ 
			6 \\ \hline
			\end{tabular}
			\begin{tabular} {|c|}
			\hline
			{\bf W}   \\ \hline
			5 \\ 
			7 \\ \hline
			\end{tabular}
		};
		\draw [dotted] ($(c1d1.south west)+(.05cm,-.05cm)$)
		rectangle ($(c1d1.north east)+(-1.65cm,.05cm)$);
		\draw [dotted] ($(c1d1.south west)+(1.4cm,-.05cm)$)
		rectangle ($(c1d1.north east)+(.00cm,.05cm)$);
		
		\node[head,fill=blue!40,right=0.83cm of sp] (so) {sorting};
		\node[fill=white!40,below left=0.15cm and -1.4cm of so] (gy) {$G \qquad  Y$};
		\node[fill=white!40,below=0.7cm of so] (g1y1) [] 
		{\fontsize{8}{9}\selectfont
			\setlength{\tabcolsep}{3pt}
			\begin{tabular} {|c|}
			\hline
			{\bf T$\downarrow$T}   \\ \hline
			7am \\ 
			8am \\ \hline
			\end{tabular}
			$\hspace*{5pt}$
			\begin{tabular}{|c|}
			\hline
			{\bf H$\downarrow$T}   \\ \hline
			6 \\ 
			8 \\ \hline
			\end{tabular}
			\begin{tabular} {|c|}
			\hline
			{\bf W$\downarrow$T}   \\ \hline
			7 \\ 
			5 \\ \hline
			\end{tabular}
		};
		\draw [dotted] ($(g1y1.south west)+(.1cm,-.05cm)$)
		rectangle	($(g1y1.north east)+(-2.05cm,.05cm)$);
		\draw [dotted] ($(g1y1.south west)+(1.15cm,-.05cm)$)
		rectangle	($(g1y1.north east)+(-.02cm,.05cm)$);
		
		\node[head,fill=white!40,right=0.7cm of so] (mo) {morphing};
		\node[fill=white!40,below=0.15cm of mo] (x) {$X$};
		\node[fill=white!40,below=0.7cm of mo] (x1) [] 
		{\fontsize{8}{9}\selectfont
			\setlength{\tabcolsep}{3pt}
			\begin{tabular} {|c|}
			\hline
			{\bf T$\downarrow$T}   \\ \hline
			7am \\ 
			8am \\ \hline
			\end{tabular}
		};
		\draw [dotted] ($(x1.south west)+(-.01cm,-.05cm)$)
		rectangle ($(x1.north east)+(.05cm,.05cm)$);
		
		\node[head,fill=white!40,right=0.4cm of mo] (sol) {eval};
		\node[fill=white!40,below =0.15cm of sol] (h) {$\quad F$};
		\node[fill=white!40,below =0.7cm of sol] (h1) [] 
		{\fontsize{8}{9}\selectfont
			\setlength{\tabcolsep}{3pt}
			\begin{tabular} {|c|}
			\hline
			{\bf H}   \\ \hline
			-0.19 \\ 
			0.31 \\ \hline
			\end{tabular}
			\begin{tabular} {|c|}
			\hline
			{\bf W}   \\ \hline
			0.27 \\ 
			-0.23 \\ \hline
			\end{tabular}
		};
		\draw [dotted] ($(h1.south west)+(.02cm,-.05cm)$)
		rectangle ($(h1.north east)+(.0cm,.05cm)$);
		
		\node[head,fill=red!40,right=0.5cm of sol] (me) {merging};
		\node[fill=white!40,below =0.15cm of me] (z) {$\quad Concat(X, F)$};
		\node[fill=white!40,below=0.68cm of me] (re) [] 
		{\fontsize{8}{9}\selectfont
			\setlength{\tabcolsep}{3pt}
			\begin{tabular} {|c|}
			\hline
			{\bf T}   \\ \hline
			7am \\ 
			8am \\ \hline
			\end{tabular}
			\begin{tabular} {|c|}
			\hline
			{\bf H}   \\ \hline
			-0.19 \\ 
			0.31 \\ \hline
			\end{tabular}
			\begin{tabular} {|c|}
			\hline
			{\bf W}   \\ \hline
			0.27 \\ 
			-0.23 \\ \hline
			\end{tabular}
		};
		\draw [dotted] ($(re.south west)+(-.0cm,-.05cm)$)
		rectangle ($(re.north east)+(.05cm,.05cm)$); 
		
		\begin{scope}[line width=2pt,gray]
		\draw ([yshift=-5pt, xshift=1.5cm]in.south) -- +(0,-2.2cm);
		\draw ([yshift=-5pt, xshift=4.6cm]in.south) -- +(0,-2.2cm);
		\draw ([yshift=-5pt, xshift=7.9cm]in.south) -- +(0,-2.2cm);
		\draw ([yshift=-5pt, xshift=10.5cm]in.south) -- +(0,-2.2cm);
		\draw ([yshift=-5pt, xshift=13.15cm]in.south) -- +(0,-2.2cm);
		\end{scope}
		\begin{scope}[every node/.style={single arrow, draw=none,fill=red!50,anchor=west,align=center}]
		\end{scope}
		\end{tikzpicture}
		\vspace{-7pt}
		\caption{Splitting, sorting, morphing, merging for query $v=\opr{\inv}{T}{}(\sigma_{T>6am}(r))$}
		\label{fig:STTree}
	\end{figure*}
	
	\begin{example}
		Figure \ref{fig:STTree} illustrates Algorithm \ref{alg:STTREE} for
		$v=\opr{\inv}{T}{}(\sigma_{T>6am}(r))$.  \emph{Splitting}: input
		list $B = (T,H,W)$ is split into order list $D = (T)$ and
		application list $ B \setminus D = (H,W)$.  \emph{Sorting}: BAT $T$ is sorted,
		producing $G$.  Then, $(H,W)$ are sorted according to $G$ returning
		$(H\downarrow T,W\downarrow T)$.  \emph{Morphing}: Since $\inv$ is
		of shape type (r$_1$,c$_1$), row contextual information is the order
		schema: $X = T\downarrow T$.  \emph{Merging}: BATs $X$ are
		concatenated with BATs $F$ to form the result.
	\end{example}
	
	\subsection{Computing the Base Result}
	\label{sec:impl-matr-oper}
	
	Line~\ref{lst:matrop} of Algorithm~\ref{alg:STTREE} calls the procedure that
	computes the matrix operation.  The computation can be done either in the kernel
	of \mbox{MonetDB} or by calling an external library (e.g., MKL~\cite{mkl}).
	Calling an external library requires copying data from BATs to the external
	format and copying the result back. The query optimizer decides about external
	library calls based on the complexity of the operation, the amount of data to be
	copied, and the relative performance of the matrix operation in \mbox{MonetDB}
	compared to the external library.
	
	The no-copy implementation of matrix operations in the kernel of
	\mbox{MonetDB} is performed over BATs directly.  Essentially, standard
	algorithms must be reduced to BAT operations.  The process of reducing
	is highly dependent on the operation.  The goal is to design
	algorithms that access entire columns and minimize accesses to single
	elements of BATs.  To achieve this standard value-based algorithms
	must be transformed to vectorized BAT operations.
	
	\begin{algorithm2e}[ht]
		\small
		\caption{$\invMA$($B$)}
		\label{alg:INV}
		$n \gets B.length;$ 
		$BR \gets IDmatrix(n)$\;
		\For{$i = 1$ \KwTo $n$}{
			$v_1 \gets sel(B_i, i);$
			$B_i \gets B_i/v_1;$ 
			$BR_i \gets BR_i/v_1$\; \label{lst:batop}
			\For{$j = 1$ \KwTo $n$}{
				\If{$i \neq j$}{
					$v_2 \gets sel(B_j,i);$
					$B_j \gets B_j - B_i * v_2$\;
					$BR_j \gets BR_j - BR_i * v_2$\;
				}
			}
		}
		\Return $BR$\;
	\end{algorithm2e}
	
	Algorithm~\ref{alg:INV} illustrates the reduction for the Gauss Jordan
	elimination method for the $\invMA$ computation.  The algorithm takes
	a list of BATs $B = (B_1, B_2, .., B_n)$ and returns the inversion as
	a list of BATs $BR$ of the same size. Function $IDmatrix(n)$ creates a
	list of BATs that represents the identity matrix of size $n \times n$.
	The selection operation $sel(B,i)$ returns the $i\textsuperscript{th}$
	value in $B$.  With the exception of the $sel$ operation, all
	operations are standard \mbox{MonetDB} BAT operations that are also
	used for relational queries.  For example, the operation on
	$B_i \gets B_i/v_1;$ divides each element of a BAT with a scalar
	value.

	\section{Performance Evaluation}
	\label{sec:Performance Evaluation}
	
	\rev{\paragraph{Setup} All runtimes are averages over 3 runs on an
		Intel(R) Xeon(R)\,E5-2603 CPU, 1.7\,GHz, 12\,cores, no
		hyper-threading, 98\,GB RAM (L1: 32+32K, L2:256K, L3:15360K),
		Debian\,4.9.189.}
	
	\paragraph{Competitors.} We empirically compare the implementations of
	our relational matrix algebra (\emph{RMA+}\footnote{\finrev{The implementation can be found here: \url{https://github.com/oksdolm/RMA}}})
	with the statistical package \emph{R}, the array database
	\emph{SciDB}, and two state-of-the-art in-database solutions,
	\emph{AIDA}~\cite{AIDA} and \emph{MADlib} \cite{MADlib}.
	(1) \rev{We implemented \emph{RMA+} in \mbox{MonetDB} (v11.29.8) with
		two options for matrix operations: (a) BATs (RMA+BAT): No-copy
		implementation in the kernel of MonetDB; (b) MKL (RMA+MKL): Copy
		BATs to an \mbox{MKL} (v2019.5.281)~\cite{mkl} compatible format
		(contiguous array of doubles), then copy the result back to BATs. We
		execute linear operations ($\add$, $\sub$, $\emu$) on BATs and use
		MKL for more complex operations. When the matrices do not fit into
		main memory we switch to BATs.  Due to the full integration of RMA+,
		\mbox{MonetDB} takes care of core usage and work distribution, and
		all cores are used for relational and for matrix operations.}
	\rev{
		(2) \emph{SciDB}~\cite{architecture_SciDB} uses an array data model, and queries are expressed in the high-level, declarative language AQL (Array Query Language)~\cite{SciDBUserGuide}. SciDB uses all available cores.}
	(3) \emph{AIDA} is a state-of-the-art solution for the integration of matrix
	operations into a relational database and was shown to outperform other
	solutions like Spark or the pandas library for Python~\cite{AIDA}.  AIDA
	executes matrix operations in Python and offers a Python-like syntax for
	relational operations, which are then translated into SQL and executed in
	\mbox{MonetDB} (v11.29.3). \rev{AIDA uses all cores both in MonetDB  and in
		Python.} We also integrate our solution into \mbox{MonetDB}, which makes AIDA a
	particularly interesting competitor.
	(4) MADlib~\cite{MADlib} (v1.10) provides a collection of UDFs for
	PostgreSQL (v9.6) for in-database matrix and statistical
	calculations. \rev{MADlib does not use multiple cores, which affects
		its overall performance.}
	(5) The \emph{R} package (v3.2.3) is highly tuned for matrix
	operations and is a representative of a non-database solution. R
	performs all relational operations with data.tables structures and
	transforms the relevant columns to matrices to compute the matrix
	operations. An alternative approach to use character matrices for all
	operations is very inefficient (cf.\ Section~\ref{sec:Data
		Transformation}). \rev{R uses all cores for matrix operations but
		runs relational operations on a single core.}
	
	\emph{Data.}  BIXI \cite{Bixi} stores trips and stations of Montreal’s
	public bicycle sharing system, years 2014-2017. DBLP \cite{DBLP}
	stores authors with their publication counts per conference as well as
	conference rankings.  The synthetic dataset used in the experiment to
	measure the effect of sparsity includes values between 0 and
	5,000,000.  All other synthetic datasets include real-valued numeric
	attributes with uniformly distributed values between 0 and 10,000.
	
	\subsection{Maintaining Contextual Information}
	\label{sec:Contextual Information}
	
	A salient feature of our approach is that contextual information is
	maintained during matrix operations.
	We analyze the scalability of maintaining context and study an
	optimization that avoids sorting.
	To this end, we generate relations with a single application column and
	an increasing number of order columns. We compute $\add$ and
	$\qqr$ on these relations. Since $\add$ and $\qqr$ are inexpensive for
	single column matrices, the main cost is the maintenance of the order
	part.
	
	To handle contextual information we split, sort, morph, and merge 
	lists of BATs (cf.\ Section~\ref{sec:RMA Integration}). Sorting
	is the most expensive operation. Fortunately, sorting is not always necessary.
	For example, permuting the input rows for the $\qqr$ operation
	will affect the order of the result rows, but will not change their values.
	Therefore, sorting is not required. In element-wise operations like $\add$,
	$\emu$, or $\sol$, only the relative order of the rows in the two input
	relations matters. Thus, only the order part of the second relation
	requires sorting (to get the same order). 
	
	Figure~\ref{fig:handling-contextual-info} shows the results. (1)
	Handling contextual information is efficient and scales to large
	numbers of attributes.  (2) The optimized operators that (partially)
	avoid sorting clearly outperform their non-optimized counterparts.
	
	\begin{figure}[htbp] \centering
		\begin{subfigure}{.48\linewidth}
			\begin{tikzpicture}[xscale=0.45, yscale=0.4]
			\begin{axis}[xlabel = \#attributes in order scpecification,
			ylabel = Runtime (sec),
			ymax = 4,
			legend pos = north west,
			x label style = {at={(axis description cs:0.5,0.01)}},
			y label style = {at={(axis description cs:.07,.5)}},
			legend style = {nodes={scale=1.2, transform shape}},
			legend cell align={left}, tick label style = {scale=1.5}, label style =
			{scale=1.5}, mark options = {solid,fill=white} ]
			\addplot [mark=diamond, dashed, mark size=3pt, line width=1pt,color=black!30!green]
			table [x=C,y=ADDM100K]{new_MDB_descriptive_part.dat};
			\addlegendentry{$\add$};
			\addplot [mark=star, mark size=3pt, line width=1pt,color=blue]
			table [x=C,y=M100K]{new_MDB_descriptive_part.dat};
			\addlegendentry{$\qqr$};
			\addplot[mark=triangle, dashed, mark size=3pt, line width=1pt,color=black!30!green]
			table [x=C, y=ADDM100KNO]{new_MDB_descriptive_part.dat};
			\addlegendentry{$\add$, relative sorting};
			\addplot[mark=square, mark size=3pt, line width=1pt,color=blue]
			table [x=C, y=M100KNO]{new_MDB_descriptive_part.dat};
			\addlegendentry{$\qqr$, w/o sorting};
			\end{axis}
			\end{tikzpicture}
			\caption{100K tuples}
			\label{fig:DPO100K}
		\end{subfigure}
		\begin{subfigure}{.48\linewidth}
			\begin{tikzpicture}[xscale=0.45, yscale=0.4]
			\begin{axis}[ xlabel = \#attributes in order schema,
			ylabel = Runtime (sec),
			ymax = 4,
			legend pos = north west,
			x label style = {at={(axis description cs:0.5,0.01)}},
			y label style = {at={(axis description cs:.08,.5)}},
			legend style = {nodes={scale=1.2, transform shape}},
			legend cell align={left},
			tick label style = {scale=1.5},
			label style = {scale=1.5},
			mark options = {solid,fill=white} ]
			\addplot[mark=diamond, dashed, mark size=3pt, line width=1pt, color=black!30!green]
			table [x=C, y=ADD1M]{MDB_desc_par_add_qqr.dat};
			\addlegendentry{$\add$};
			\addplot[mark=triangle, dashed, mark size=3pt, line width=1pt, black!30!green]
			table [x=C, y=ADD1MRO]{MDB_desc_par_add_qqr.dat};
			\addlegendentry{$\add$, relative sorting};
			\addplot [mark=star, mark size=3pt, line width=1pt, color=blue]
			table [x=C,y=QQR1M]{MDB_desc_par_add_qqr.dat};
			\addlegendentry{$\qqr$};
			\addplot [mark=square, mark size=3pt, line width=1pt, color=blue]
			table [x=C,y=QQR1MNO]{MDB_desc_par_add_qqr.dat};
			\addlegendentry{$\qqr$, w/o sorting};
			\end{axis}
			\end{tikzpicture}
			\caption{1M tuples}
			\label{fig:DPO1M}
		\end{subfigure}
		\caption{Handling contextual information}
		\label{fig:handling-contextual-info}
	\end{figure}
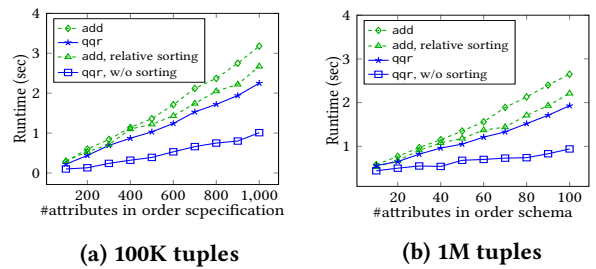
	
	Note that a number of operations ($\cpd$, $\sol$, $\rqr$, $\dsd$,
	$\tra$, $\det$, $\rnk$) do not preserve row context since the number
	of rows changes.
	Instead, a single column with predefined values (operation name or
	attribute names of the application schema) is created, which is
	negligible in the overall runtime.
	
	\subsection{Wide and Sparse Relations}
	\label{sec:Relations with Many Attributes and Sparse Relations}
	
	\rev{\paragraph{Wide relations.} Current databases scale better in the
		number of tuples than in the number of attributes. We test our RMA+
		implementation in MonetDB on wide relations. We generate relations
		with 1000 tuples, one order attribute, and a varying number of
		application attributes. In Table~\ref{tbl:ADD10K}, we increase the
		number of attributes from 1K to 10K and measure the runtime of the
		$\add$ operation. MonetDB can handle wide relations with several
		thousands of attributes, even though the runtime per column
		increases with the attribute number.
		
		\begin{table}[ht] \centering \setlength{\tabcolsep}{3pt}
			\caption{ $\add$ over wide relations in RMA+}
			\label{tbl:ADD10K}
			\scriptsize\vspace{-5pt}
			\begin{tabular} {|c|c|c|c|c|c|c|c|c|c|c|} \hline%
				{\bf \#attr} & \bf 1K & \bf 2K & \bf 3K & \bf 4K & \bf 5K & \bf 6K & \bf 7K & \bf 8K & \bf 9K & \bf 10K
				\\ \hline {\bf sec} & 0.6 & 2.2 & 4.8 & 8.8 & 13.4 & 20 & 27 & 36 & 47 & 62 \\
				\hline
			\end{tabular}
		\end{table}
		
		\paragraph{Sparse relations.} We analyze the effect of MonetDB's
		built-in compression on relations with many zeros. We add two
		relations (5M tuples, one order, 10 application attributes) with
		uniformly distributed non-zero values (range 1-5M). In
		Table~\ref{tbl:ADDSPARSE} we increase the percentage of zero values
		(position of zeros is random) and measure the runtime: The $\add$
		operation on sparse matrices is up to two times faster than the same
		operation on dense matrices.  Thus, RMA+ leverages MonetDB's
		compression features.
		
		\begin{table}[ht] \centering
			\setlength{\tabcolsep}{2pt}
			\caption{$\add$ over sparse relations in RMA+}
			\label{tbl:ADDSPARSE}
			\scriptsize\vspace{-5pt}
			\begin{tabular} {|c|c|c|c|c|c|c|c|c|c|c|c|} \hline%
				{\bf \%} &  \bf 0 & \bf 10 & \bf 20 & \bf 30 & \bf 40 & \bf 50 & \bf 60 & \bf 70 & \bf 80 & \bf 90 & \bf 100 \\ \hline
				{\bf sec} & 1.68 & 1.60 & 1.49  & 1.41 & 1.33 & 1.25 & 1.16 & 0.99 & 0.94 & 0.89 & 0.76 \\ \hline 
			\end{tabular}
		\end{table}
		
	}
	
	\subsection{RMA+ vs.\ Non-Database Approaches}
	\label{sec:R and MonetDB kernel relative performance}
	
	We study the scalability of RMA+ to large relations and compare to $R$
	as a non-database solutions for matrix operations. In Table
	\ref{tbl:5M50M} we measure the runtime for $\qqr$ on tables with
	\rev{up to 100M tuples} and 70 attributes in the application
	schema. \rev{For relations up to a size of 50Mx40, RMA+ delegates the
		matrix computation to MKL; the runtime includes copying the
		data. RMA+ is consistently faster than R since MKL can better
		leverage the hardware. R fails for sizes above 50Mx40 since it runs
		out of memory. In RMA+ we switch to the BAT implementation, which
		leverages the memory management of MonetDB. The Gram-Schmidt $\qqr$
		baseline~\cite{qr_alg} that we implemented over BATs is slower than
		the MKL algorithm (e.g., 834 vs.\ 61.4\,sec for 50Mx40), which
		explains the increase in runtime. RMA+ scales to large relations
		that do not fit into memory (e.g., relation size 100Mx70 requires
		56GB).}

	\begin{table}[htbp] \centering
		\caption{Runtimes of $\qqr$ in seconds in R and RMA+}
		\label{tbl:5M50M}
		\scriptsize\vspace{-5pt}
		\begin{tabular} {|c|c|c|c|c|c|c|} \hline%
			{} & \multicolumn{2}{c|}{\bf 10 attr}
			& \multicolumn{2}{c|}{\bf 40 attr} & \multicolumn{2}{c|}{\bf 70 attr} \\
			\hline
			{\bf System} & \bf R & \bf RMA+ & \bf R & \bf RMA+ & \bf R & \bf RMA+ \\ \hline
			{\bf 5M tup} & 3.5 & 2.1  & 20 & 6.6 & 47 & 11.6 \\ \hline 
			{\bf 50M tup} & 37 & 21.3 & 221 & 61.4 & fail  & 2018 \\ \hline
			{\bf 100M tup} & 74 & 40 & fail & 1690 & fail  & 4064 \\ \hline
		\end{tabular}
	\end{table}
	
	\subsection{RMA+ vs.\ Array Databases}
	
	We study the performance of RMA+ vs.\
	SciDB~\cite{architecture_SciDB} as a representative of array
	databases.  We compute $\add$ on two matrices with 10 columns and a
	varying number of rows, followed by a selection
	\footnote{We run this experiment on Ubuntu 14.04 since SciDB does
		not support Debian; Ubuntu runs on a server with 4 cores and 16GB
		of RAM.}.
	The resulting runtimes for Ubuntu are shown in Table~\ref{tbl:RMA+SciDB}. RMA+ outperforms SciDB by more than an order of magnitude. RMA+ performs addition directly over pairs of relations, while SciDB must compute a so-called array join~\cite{SciDBUserGuide} over the input arrays in order to add their
	values.
	
	\setlength{\tabcolsep}{3pt}
	\begin{table}[htbp] \centering
		\caption{$\add$ followed by a selection: RMA+ vs.\ SciDB}
		\vspace{-5pt}
		\label{tbl:RMA+SciDB}
		\scriptsize
		\begin{tabular} {|c|c|c|c|c|} \hline%
			{\bf \#tuples} & \bf 1M & \bf 5M & \bf 10M & \bf 15M  \\ \hline
			{\bf RMA+} & 4.6s & 24.4s  & 1m18s & 1m39s  \\ \hline 
			{\bf SciDB} & 1m21s & 7m6s  & 13m2s & 18m23s  \\ \hline 
		\end{tabular}
	\end{table}

	\subsection{Overhead of Data Transformation}
	\label{sec:Data Transformation}
	
	We investigate the overhead of data transformation for various matrix
	operations in a mixed relational/matrix scenario.
	
	\rev{RMA+ is free to execute matrix operations directly on BATs or rearrange the
		numerical data in main memory and delegate the matrix operations to specialized
		packages like MKL~\cite{mkl}.}
	R does not enjoy this flexibility: R uses the matrix data type for
	matrix operations and the data.tables storage structure for relational
	operations.  While data.tables supports simple linear operations like
	linear model construction, the data must be transformed to the matrix
	type for more complex operations like $\cpdMA$, $\opdMA$, or
	$\mmuMA$. Matrices cannot store a mix of numerical and non-numerical
	values, which is required when working with tables; R offers character
	matrices, but they are very inefficient, e.g., joining trips and
	stations in the BIXI dataset takes 40\,sec for the character matrix type
	and less than 2\,sec
	for data.tables. 
	
	Figure~\ref{fig:COFPDTR} shows the percentage of time spent for data
	transformations on relations with 50 columns and a varying number of
	rows (100k to 500k). For R we measure the time of transforming the
	relation from data.table to matrix and back as a percentage of the
	overall query time, which includes the actual matrix operation.  For
	RMA+ we measure the time share for copying the data from a list of
	BATs to a contiguous, one-dimensional array for MKL, and for copying
	the result back; the overall runtime in addition includes the matrix
	computation in MKL (but excludes the \mbox{MonetDB} query pipeline of
	query parsing, query tree creation, etc.).
	
	\begin{figure}[!htb] \centering \setlength{\tabcolsep}{1.2pt}
		{\fontsize{8}{9}\selectfont
			\begin{subfigure}{0.48\linewidth}
				\begin{tabular} {r|c|c|c|c|c|c|}
					\multicolumn{1}{l|}{\#rows} & \multicolumn{6}{c}{(\#columns = 50)} \\
					\hline	
					500K & \cellcolor{gray!81}81   & \cellcolor{gray!75}75 & \cellcolor{gray!64}64 & \cellcolor{gray!21}21 & \cellcolor{gray!7}7 & \cellcolor{gray!7}7 \\ \hline
					300K &\cellcolor{gray!79}79 & \cellcolor{gray!77}77  & \cellcolor{gray!63}63 & \cellcolor{gray!21}21 & \cellcolor{gray!7}7 & \cellcolor{gray!7}7 \\ \hline
					100K & \cellcolor{gray!84}84 & \cellcolor{gray!74}74 & \cellcolor{gray!69}69 & \cellcolor{gray!23}23 & \cellcolor{gray!9}9 & \cellcolor{gray!10}10 \\ \hline
					& $\addMA$ & $\emuMA$ & $\mmuMA$ & $\qqrMA$ & $\dsdMA$ & $\vsdMA$ 
				\end{tabular}
				\caption{Data.table and matrix}\label{fig:COPYR}
			\end{subfigure}
			\begin{subfigure}{0.48\linewidth}
				\begin{tabular} {r|c|c|c|c|c|c|}
					\multicolumn{1}{l|}{\#rows} & \multicolumn{6}{c}{(\#columns = 50)} \\
					\hline	
					500K & \cellcolor{gray!92}92  & \cellcolor{gray!92}92
					& \cellcolor{gray!86}86  & \cellcolor{gray!53}53
					& \cellcolor{gray!44}44  & \cellcolor{gray!43}43 \\ \hline
					300K  & \cellcolor{gray!91}91  & \cellcolor{gray!91}91
					& \cellcolor{gray!86}86  & \cellcolor{gray!55}55 
					& \cellcolor{gray!45}45  & \cellcolor{gray!40}40 \\ \hline
					100K  & \cellcolor{gray!86}86  & \cellcolor{gray!86}86
					& \cellcolor{gray!80}80  & \cellcolor{gray!48}48
					& \cellcolor{gray!37}37  & \cellcolor{gray!35}35 \\ \hline
					& $\addMA$ & $\emuMA$ & $\mmuMA$ & $\qqrMA$ & $\dsdMA$ & $\vsdMA$
				\end{tabular}
				\caption{List of BATs and 1D array}\label{fig:RMALCOPY}
			\end{subfigure}
		}
		\vspace{-5pt}
		\caption{Data transformation share: (a) R, (b) RMA+}
		\label{fig:COFPDTR}
	\end{figure}
	
	Clearly, the overhead of transforming data matters for both R and
	RMA+. We draw the following conclusions: 
	(a) Transforming data between data structures is costly.
	(b) For simple operations like $\addMA$ and $\emuMA$, the transformation
	overhead dominates the overall runtime \rev{(up to 92\%)}.
	(c) For complex operations,  the performance of the matrix operation dominates
	the overall runtime.
	
	\subsection{Efficiency for Mixed Workloads}
	\label{sec:mixed-workloads}
	
	We analyze four workloads that require a mix of relational operations and matrix
	operations, and we compare our implementation of RMA (RMA+) to its competitors
	(R, AIDA, MADlib). The workloads stem from applications on our real-world
	datasets and differ in the complexity of relational vs.\ matrix part. On the
	BIXI dataset, we compute (1) the linear regression between distance and duration
	for individual trips, and (2) journeys connecting up to 5 trips; on DBLP we
	compute the (3) covariance between conferences based on the publication counts
	per conference and author; (4) on a synthetic dataset based on BIXI we
	count trips per rider.

	\paragraph{(1) Trips -- Ordinary Linear Regression}
	
	Trips in BIXI include start date and start station, end date and end station,
	duration, and a membership flag for the rider; stations have a code, a name, and
	coordinates. At the level of relations, we need to perform the following data
	preparation steps: (a) Aggregate the trips and select those trips that were
	performed at least 50 times; (b) join trips and stations to retrieve the station
	coordinates and compute the distance. We use the OLS method~\cite{LLS} to
	compute the linear regression between distance and duration.  OLS uses cross
	product, matrix multiplication, and inversion:
	$\mmuMA(\invMA(\cpdMA(A,A)),\cpdMA(A,V))$, where $A$ is the matrix with the
	independent variables, and $V$ is the vector with the dependent variable.
	
	Figure \ref{fig:LinReg_systems} shows the runtime results for trips reported
	in the years 2014 (3.1M trips), 2014-2015 (6.1M trips), 2014-2016
	(10.5M trips), and 2014-2017 (14.5M trips), respectively. \rev{The
		input data consists of numeric and non-numeric types such as date
		and time.}  We break the runtime down into data preparation (solid
	area of the bar) and matrix computation time (dashed light area) for
	RMA+, R, and AIDA; \rev{for R we also show the load time from a CSV
		file (dark area). RMA+ and AIDA outperform R and MADlib in all
		scenarios. R performs poorly on the relational operations of the
		data preparation step: The join implementation of R does not
		leverage multiple cores, and R lacks a query optimizer, which
		adversely affects the relational performance.
		MADlib is outperformed by all other solutions due to the slow
		computation of the linear regression.
		RMA+ outperforms AIDA on all datasets. Although both RMA+ and AIDA compute the
		relational operations in MonetDB, RMA+ is up to 6.3  times faster: While AIDA passes pointers to access numerical Python data in
		MonetDB, this does not work for other data types (e.g., date, time,  string) due
		to different storage formats~\cite{AIDA_SSDBM}. Therefore, expensive data
		transformations must be applied.} 
	
	\begin{figure}[htbp] \centering	
		\begin{subfigure}{.48\linewidth}
			\begin{tikzpicture}[
			xscale=0.35,
			yscale=0.3,
			every axis/.style={
				width=11cm,
				ybar stacked,
				ymin=0,ymax=27,
				symbolic x coords={3.1,6.5,10.5,14.5},
				bar width=9pt,
				xtick={3.1,6.5,10.5,14.5},  
				label style={font=\Huge},
				tick label style={font=\Huge}
			},
			]
			
			\begin{axis}
			[bar shift=-15pt,
			hide axis,
			legend style={font=\Huge},
			legend style={at={(0.32,0.98)}},
			legend style={draw=none}]
			\addplot+[draw=none,fill=cyan]
			table [x=C,y=RMA_join_c]{BS_unlinreg_RMA.dat};
			\addlegendentry{RMA+};
			\addplot+[draw=none,fill=cyan!50!white,
			postaction={pattern=north east lines}]
			table [x=C,y=RMA_lapack_separ]{BS_unlinreg_RMA.dat};
			\addplot+[draw=none,fill=cyan!50!white,
			postaction={pattern=north east lines}]
			table [x=C,y=MKL]{BS_unlinreg_MATH_Pack.dat};
			\end{axis}
			
			\begin{axis}[bar shift=-5pt, 
			hide axis,
			legend style={font=\Huge},
			legend style={at={(0.55,0.98)}},
			legend style={draw=none}]
			
			\addplot+[draw=none,fill=red]
			table [x=C,y=AIDA_join]{BS_unlinreg_AIDA.dat};
			\addlegendentry{AIDA};
			\addplot+[draw=none,fill=red!50!white,
			postaction={pattern=north east lines}]
			table [x=C,y=AIDA_linreg]{BS_unlinreg_AIDA.dat};
			\end{axis}
			
			\begin{axis}[bar shift=5pt,
			legend style={at={(0.65,0.98)}},
			legend style={font=\Huge},
			legend style={draw=none},
			xlabel = \#tuples (M),
			ylabel = Runtime (sec),
			label style = {font=\Huge},
			y label style = {at={(axis description cs:.01,.5)}}]
			\addplot+[forget plot,draw=none,fill=black!60!green]
			table [x=C,y=R_load]{BS_unlinreg_R.dat};
			\addplot+[draw=none,fill=black!30!green]
			table [x=C,y=R_join]{BS_unlinreg_R.dat};
			\addlegendentry{R};
			\addplot+[draw=none,fill=black!30!green!50!white,
			postaction={pattern=north east lines}]
			table [x=C,y=R_OLS]{BS_unlinreg_R.dat};
			\end{axis}

			\begin{axis}
			[bar shift=15pt,
			hide axis,
			legend style={font=\Huge},
			legend style={at={(0.95,0.98)}},
			legend style={draw=none}]
			\addplot+[draw=none,fill=black!30!brown]
			table [x=C,y=MADlib_join]{BS_unlinreg_MADlib.dat};
			\addlegendentry{MADlib};
			\addplot+[draw=none, fill=black!30!brown!50!white,
			postaction={pattern=north east lines}]
			table [x=C,y=MADlib_linreg]{BS_unlinreg_MADlib.dat};
			\end{axis}
			
			\end{tikzpicture}
			\caption{System Comparison}\label{fig:LinReg_systems}
		\end{subfigure}
		\begin{subfigure}{.48\linewidth}
			\begin{tikzpicture}[
			xscale=0.35,
			yscale=0.3,
			every axis/.style={
				width=11cm,
				ybar stacked,
				ymin=0,ymax=5,
				symbolic x coords={3.1,6.5,10.5,14.5},
				bar width=9pt,
				xtick={3.1,6.5,10.5,14.5}, 
				label style={font=\huge},
				tick label style={font=\huge} 
			},
			]
			
			\begin{axis}
			[bar shift=-5pt,
			hide axis,
			legend style={font=\huge},
			legend style={at={(0.45,0.98)}},
			legend style={draw=none}]
			\addplot+[draw=none,fill=cyan]
			table [x=C,y=RMA_join_c]{BS_unlinreg_RMA.dat};
			\addlegendentry{RMA+MKL};
			\addplot+[draw=none,fill=cyan!50!white,
			postaction={pattern=north east lines}]
			table [x=C,y=RMA_lapack_separ]{BS_unlinreg_RMA.dat};
			\addplot+[draw=none,fill=cyan!50!white,
			postaction={pattern=north east lines}]
			table [x=C,y=MKL]{BS_unlinreg_MATH_Pack.dat};
			\end{axis}
			
			\begin{axis}[bar shift=5pt,
			legend style={font=\Huge},
			legend style={at={(0.85,0.98)}},
			legend style={draw=none},
			xlabel = \#tuples (M),
			ylabel = Runtime (sec),
			label style = {font=\Huge},
			y label style = {at={(axis description cs:.05,.5)}}]
			\addplot+[draw=none,fill=black]
			table [x=C,y=RMA_join_c]{BS_unlinreg_RMA.dat};
			\addlegendentry{RMA+BAT};
			\addplot+[draw=none,fill=black!50!white,
			postaction={pattern=north east lines}]
			table [x=C,y=RMA_linreg_separ]{BS_unlinreg_RMA.dat};
			\end{axis}
			\end{tikzpicture}
			\caption{RMA+BAT vs RMA+MKL}\label{fig:LinReg_BATvsMKL}
		\end{subfigure}
		\vspace{-5pt}
		\caption{Trips (Ordinary Linear Regression)}\label{fig:LinReg}\vspace{-10pt}
	\end{figure}
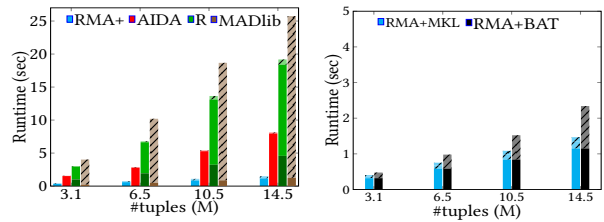

	\paragraph{(2) Journeys -- Multiple Linear Regression}
	\label{sec:mult-line-regr}
	
	We compose trips that meet in a station into journeys. \rev{We start
		from 15M one-trip journeys of the form (start station, end station,
		duration); all attributes are numerical. During data preparation, we
		perform joins to create journeys of up to five trips, select those
		that appear at least 50 times, and join stations with their
		coordinates to compute the distances between subsequent stations in
		a journey.} At the matrix level, we do a multiple linear regression
	analysis with the distances as independent variables and the overall
	duration as the dependent variable.
	
	Figure \ref{fig:MulReg_systems} shows the runtime for journey lengths of 1
	to 5 trips (i.e., 1 to 5 independent variables).  The solid part is
	the time for data preparation (relational operations); the dashed
	light part is the time for multiple linear regression (matrix
	operations).
	\rev{RMA+ and AIDA again outperform R on the relational part of the
		query. The relational part operates on purely numerical data and
		AIDA shows comparable join performance to RMA+. MADlib spends about two
		third of the relational runtime on distance computations and is
		therefore slower than its competitors also on the relational part.}
	
	\begin{figure}[htbp] \centering	
		\begin{subfigure}{.48\linewidth}
			\begin{tikzpicture}[
			xscale=0.35,
			yscale=0.3,
			every axis/.style={
				width=11cm,
				ybar stacked,
				ymin=0,ymax=215,
				symbolic x coords={1,2,3,4,5},
				bar width=9pt,
				xtick={1,2,3,4,5},  
				label style={font=\Huge},
				tick label style={font=\Huge} 
			},
			]
			
			\begin{axis}
			[bar shift=-15pt,
			hide axis,
			legend style={font=\Huge},
			legend style={at={(0.32,0.98)}},
			legend style={draw=none}]
			\addplot+[draw=none,fill=cyan]
			table [x=C,y=RMA_join_c]{BS_mlinreg_RMA.dat};
			\addlegendentry{RMA+};
			\addplot+[draw=none,fill=cyan!50!white,
			postaction={pattern=north east lines}]
			table [x=C,y=RMA_lapack_separ]{BS_mlinreg_RMA.dat};
			\addplot+[draw=none,fill=cyan!50!white,
			postaction={pattern=north east lines}]
			table [x=C,y=MKL]{BS_mlinreg_MATH_Pack.dat};
			\end{axis}

			\begin{axis}
			[bar shift=-5pt,
			hide axis,
			legend style={font=\Huge},
			legend style={at={(0.55,0.98)}},
			legend style={draw=none}]
			\addplot+[draw=none,fill=red]
			table [x=C,y=AIDA_join]{BS_mlinreg_AIDA.dat};
			\addlegendentry{AIDA};
			\addplot+[draw=none,fill=red!50!white,
			postaction={pattern=north east lines}]
			table [x=C,y=AIDA_linreg]{BS_mlinreg_AIDA.dat};
			\end{axis}

			\begin{axis}[bar shift=5pt,
			legend style={font=\Huge},
			legend style={at={(0.65,0.98)}},
			legend style={draw=none},
			xlabel = \#trips,
			ylabel = Runtime (sec),
			label style = {font=\Huge},
			y label style = {at={(axis description cs:-.02,.5)}}]
			\addplot+[forget plot,draw=none,fill=black!60!green]
			table [x=C,y=R_load]{BS_mlinreg_R.dat};
			\addplot+[draw=none,fill=black!30!green]
			table [x=C,y=R_join]{BS_mlinreg_R.dat};
			\addlegendentry{R};
			\addplot+[draw=none,fill=black!30!green!50!white,
			postaction={pattern=north east lines}]
			table [x=C,y=R_OLS]{BS_mlinreg_R.dat};
			\end{axis}
			
			\begin{axis}
			[bar shift=15pt,
			hide axis,
			legend style={font=\Huge},
			legend style={at={(0.95,0.98)}},
			legend style={draw=none}]
			\addplot+[draw=none,fill=black!30!brown]
			table [x=C,y=MADlib_join]{BS_mlinreg_MADlib.dat};
			\addlegendentry{MADlib};
			\addplot+[draw=none,fill=black!30!brown!50!white,
			postaction={pattern=north east lines}]
			table [x=C,y=MADlib_linreg]{BS_mlinreg_MADlib.dat};
			\end{axis}
			
			\end{tikzpicture}
			\caption{System Comparison}\label{fig:MulReg_systems}
		\end{subfigure}
		\begin{subfigure}{.48\linewidth}
			\begin{tikzpicture}[
			xscale=0.35,
			yscale=0.3,
			every axis/.style={
				width=11cm,
				ybar stacked,
				ymin=0,ymax=10,
				symbolic x coords={1,2,3,4,5},
				bar width=9pt,
				xtick={1,2,3,4,5},  
				label style={font=\Huge},
				tick label style={font=\Huge}  
			},
			]
			
			\begin{axis}
			[bar shift=-5pt,
			hide axis,
			legend style={font=\huge},
			legend style={at={(0.45,0.98)}},
			legend style={draw=none}]
			\addplot+[draw=none,fill=cyan]
			table [x=C,y=RMA_join_c]{BS_mlinreg_RMA.dat};
			\addlegendentry{RMA+MKL};
			\addplot+[draw=none,fill=cyan!50!white,
			postaction={pattern=north east lines}]
			table [x=C,y=RMA_lapack_separ]{BS_mlinreg_RMA.dat};
			\addplot+[draw=none,fill=cyan!50!white,
			postaction={pattern=north east lines}]
			table [x=C,y=MKL]{BS_mlinreg_MATH_Pack.dat};
			\end{axis}

			\begin{axis}
			[bar shift=5pt,
			legend style={font=\Huge},
			legend style={at={(0.85,0.98)}},
			legend style={draw=none},
			xlabel = \#trips,
			ylabel = Runtime (sec),
			label style = {font=\Huge},
			y label style = {at={(axis description cs:.03,.5)}}]
			\addplot+[draw=none,fill=black]
			table [x=C,y=RMA_join_c]{BS_mlinreg_RMA.dat};
			\addlegendentry{RMA+BAT};
			\addplot+[draw=none,fill=black!50!white,
			postaction={pattern=north east lines}]
			table [x=C,y=RMA_linreg_separ]{BS_mlinreg_RMA.dat};
			\end{axis}

			\end{tikzpicture}
			\caption{RMA+BAT vs RMA+MKL}\label{fig:MulReg_BATvsMKL}
		\end{subfigure}
		\vspace{-5pt}
		\caption{Journeys (Multiple Linear Regression)}\label{fig:MulReg}\vspace{-10pt}
	\end{figure}
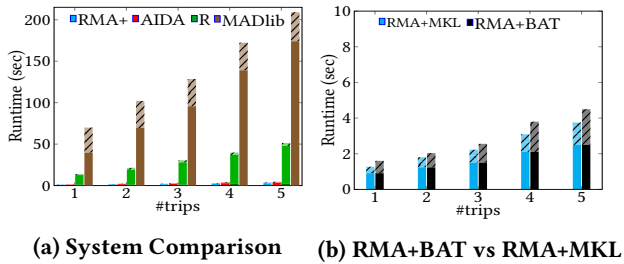

	\paragraph{(3) Conferences -- Covariance Computation}
	\label{sec:covariance}
	
	We compute the covariance between conferences with A++ rating to lower rated
	conferences based on the number of publications per author and conference. The
	data includes two tables: $ranking$ stores a rating (e.g., A++, A+, B) for each
	conference. $publication$ stores the number of publications per author and
	conference; the first attribute is the author, the other attributes are
	conference names (i.e., the result of SQL PIVOT over a count-aggregate by
	conference and author). The query computes the covariance matrix on
	$publication$ and joins the result with $ranking$ to select A++ conferences.
	
	We measure the runtime  for $publication$ tables of increasing sizes: (1)
	337363x266 (i.e., 337363 authors and 266 conferences), (2) 550085x519, (3)
	722891x744, and (4) 876559x882. The $ranking$ table stores 882 tuples. Note that
	the number of result rows of covariance is identical to the number of input
	columns, e.g., covariance of $publications$ with 266 columns  returns a relation
	(or matrix) of size 266x266.
	
	Figure \ref{fig:CovConf_systems} shows the runtime results for RMA+,
	R, and AIDA. \rev{MADlib runs for 77, 429, 1086, resp.\ 1814 seconds
		on the different relation sizes and, thus, is omitted from the
		figure.} In all systems, the covariance computation dominates the
	overall runtime with at least 90\%.  \rev{Since AIDA does not support
		covariance, we implement covariance via cross product~\cite{SamCov}
		in all algorithms except MADlib, which has a \texttt{cov()} function
		but does not support cross product. For the cross product in RMA+ we
		use \finrev{the routine \texttt{cblas\_dsyrk()} since the result of
			multiplication is symmetric}, in AIDA we use \texttt{a.t\,@\,a},
		in R we use \texttt{crossproduct} \footnote{\finrev{We do not use
				cov() function since it uses a single core only and is
				slower.}}. }
	
	\rev{Note that the covariance computations in AIDA and R do not return
		contextual information. In order to join the result with $ranking$
		and to select all A++ conferences, the conference names must be
		manually added as a new column.}
	
	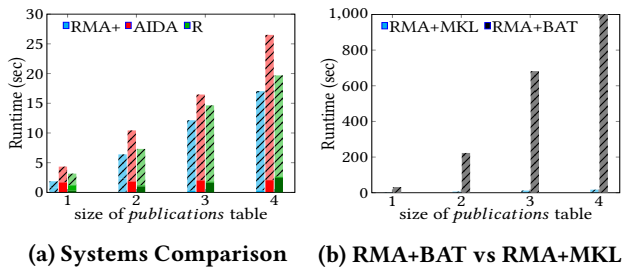
\begin{figure}[htbp] \centering	
		\begin{subfigure}{.48\linewidth}
			\begin{tikzpicture}[
			xscale=0.35,
			yscale=0.3,
			every axis/.style={
				width=11cm,
				ybar stacked,
				ymin=0,ymax=30,
				symbolic x coords={1,2,3,4},
				bar width=9pt,
				xtick={1,2,3,4}, 
				label style={font=\Huge},
				tick label style={font=\Huge} 
			},
			]
			
			\begin{axis}
			[bar shift=-15pt,
			hide axis,
			legend style={font=\Huge},
			legend style={at={(0.32,0.98)}},
			legend style={draw=none}]
			\addplot+[draw=none,fill=cyan]
			table [x=C,y=RMA_join]{BS_cov_RMA.dat};
			\addlegendentry{RMA+};
			\addplot+[draw=none,fill=cyan!50!white,
			postaction={pattern=north east lines}]
			table [x=C,y=RMA_lapack]{BS_cov_RMA.dat};
			\addplot+[draw=none,fill=cyan!50!white,
			postaction={pattern=north east lines}]
			table [x=C,y=MKL_cov_dsyrk]{BS_cov_RMA.dat};
			\end{axis}
			
			\begin{axis}[	
			bar shift=-5pt,
			hide axis,
			legend style={font=\Huge},
			legend style={at={(0.55,0.98)}},
			legend style={draw=none}]
			\addplot+[draw=none,fill=red]
			table [x=C,y=AIDA_join]{BS_cov_AIDA.dat};
			\addlegendentry{AIDA};
			\addplot+[draw=none,fill=red!50!white,
			postaction={pattern=north east lines}]
			table [x=C,y=AIDA_cov]{BS_cov_AIDA.dat};
			\end{axis}

			\begin{axis}[bar shift=5pt,
			legend style={font=\Huge},
			legend style={at={(0.65,0.98)}},
			legend style={draw=none},
			xlabel = size of $publications$ table,
			ylabel = Runtime (sec),
			label style = {font=\Huge},
			y label style = {at={(axis description cs:.01,.5)}}]
			\addplot+[forget plot,draw=none,fill=black!60!green]
			table [x=C,y=R_load]{BS_cov_R.dat};
			\addplot+[draw=none,fill=black!30!green]
			table [x=C,y=R_join]{BS_cov_R.dat};
			\addlegendentry{R};
			\addplot+[draw=none,fill=black!30!green!50!white,
			postaction={pattern=north east lines}]
			table [x=C,y=R_cov]{BS_cov_R.dat};
			\end{axis}
			
			\end{tikzpicture}
			\caption{Systems Comparison}
			\label{fig:CovConf_systems}
		\end{subfigure}
		\begin{subfigure}{.48\linewidth}
			\begin{tikzpicture}[
			xscale=0.35,
			yscale=0.3,
			every axis/.style={
				width=11cm,
				ybar stacked,
				ymin=0,ymax=1000,
				symbolic x coords={1,2,3,4},
				bar width=9pt,
				xtick={1,2,3,4},  
				label style={font=\Huge},
				tick label style={font=\Huge}
			},
			]
			
			\begin{axis}
			[bar shift=-5pt,
			hide axis,
			legend style={font=\Huge},
			legend style={at={(0.45,0.98)}},
			legend style={draw=none}]
			\addplot+[draw=none,fill=cyan]
			table [x=C,y=RMA_join]{BS_cov_RMA.dat};
			\addlegendentry{RMA+MKL};
			\addplot+[draw=none,fill=cyan!50!white,
			postaction={pattern=north east lines}]
			table [x=C,y=RMA_lapack]{BS_cov_RMA.dat};
			\addplot+[draw=none,fill=cyan!50!white,
			postaction={pattern=north east lines}]
			table [x=C,y=MKL_cov_dsyrk]{BS_cov_RMA.dat};
			\end{axis}

			\begin{axis}[bar shift=5pt,
			legend style={font=\Huge},
			legend style={at={(0.85,0.98)}},
			legend style={draw=none},
			xlabel = size of $publications$ table,
			ylabel = Runtime (sec),
			label style = {font=\Huge},
			y label style = {at={(axis description cs:-.04,.5)}}]
			\addplot+[draw=none,fill=black]
			table [x=C,y=RMA_join]{BS_cov_RMA.dat};
			\addlegendentry{RMA+BAT};
			\addplot+[draw=none,fill=black!50!white,
			postaction={pattern=north east lines}]
			table [x=C,y=RMA_cov]{BS_cov_RMA.dat};
			\end{axis}
			
			\end{tikzpicture}
			\caption{RMA+BAT vs RMA+MKL}
			\label{fig:CovConf_BATvsMKL}
		\end{subfigure}
		\caption{Conferences (Covariance Computation)}\label{fig:CovConf}\vspace{-15pt}
	\end{figure}
	
	\paragraph{(4) Trip Count}
	\label{sec:tripcount}
	
	\rev{In Figure \ref{fig:ADD} we compute the number of trips per rider
		to 10 different destinations. Each tuple in the input relations
		stores a rider and the number of trips to each of the 10 locations
		for one year. We use $\add$ on the relations of two different years
		to get the trip count for a period of two years. We vary the number
		of riders from 1M to 15M and measure the runtime.  Since $\add$ is a
		simple operation, RMA+ uses the no-copy implementation on BATs
		(RMA+BAT). RMA+ is faster than AIDA and R because it does not
		transfer data to Python (as AIDA) and does not translate data.tables
		to matrices (as R). MADlib takes 23, 119, 299, resp.\ 480 seconds
		for the different input sizes and, thus, is again omitted from the
		figure.}

	\begin{figure}[htbp] \centering	
		\begin{subfigure}{.48\linewidth}
			\begin{tikzpicture}[
			xscale=0.35,
			yscale=0.3,
			every axis/.style={
				width=11cm,
				ybar stacked,
				ymin=0,ymax=7,
				symbolic x coords={1, 5, 10, 15},
				bar width=9pt,
				xtick={1, 5, 10, 15},  
				label style={font=\Huge},
				tick label style={font=\Huge}
			},
			]
			
			\begin{axis}
			[bar shift=-15pt,
			hide axis,
			legend style={font=\Huge},
			legend style={at={(0.32,0.98)}},
			legend style={draw=none}]
			\addplot+[draw=none,fill=black!50!white,
			postaction={pattern=north east lines}]
			table [x=C,y=RMA+BAT]{BS_add_ALLSYS.dat};
			\addlegendentry{RMA+};;
			\end{axis}
			
			\begin{axis}[	
			bar shift=-5pt,
			hide axis,
			legend style={font=\Huge},
			legend style={at={(0.55,0.98)}},
			legend style={draw=none}]
			\addplot+[draw=none,fill=red!50!white,
			postaction={pattern=north east lines}]
			table [x=C,y=AIDA]{BS_add_ALLSYS.dat};
			\addlegendentry{AIDA};
			\end{axis}

			\begin{axis}[bar shift=5pt,
			legend style={font=\Huge},
			legend style={at={(0.65,0.98)}},
			legend style={draw=none},
			xlabel = \#tuples(M),
			ylabel = Runtime (sec),
			label style = {font=\Huge},
			y label style = {at={(axis description cs:.02,.5)}}]
			\addplot+[forget plot,draw=none,fill=black!60!green]
			table [x=C,y=R_load]{BS_add_ALLSYS.dat};
			\addplot+[draw=none,fill=black!30!green!50!white,
			postaction={pattern=north east lines}]
			table [x=C,y=R]{BS_add_ALLSYS.dat};
			\addlegendentry{R};
			\end{axis}
			
			\end{tikzpicture}
			\caption{Systems Comparison}
			\label{fig:ADD_systems}
		\end{subfigure}
		\begin{subfigure}{.48\linewidth}
			\begin{tikzpicture}[
			xscale=0.35,
			yscale=0.3,
			every axis/.style={
				width=11cm,
				ybar stacked,
				ymin=0,ymax=7,
				symbolic x coords={1, 5, 10, 15},
				bar width=9pt,
				xtick={1, 5, 10, 15},  
				label style={font=\Huge},
				tick label style={font=\Huge}
			},
			]
			
			\begin{axis}
			[bar shift=-5pt,
			hide axis,
			legend style={font=\Huge},
			legend style={at={(0.45,0.98)}},
			legend style={draw=none}]
			\addplot+[draw=none,fill=cyan!50!white,
			postaction={pattern=north east lines}]
			table [x=C,y=RMA+MKL]{BS_add_ALLSYS.dat};
			\addlegendentry{RMA+MKL};
			\end{axis}
			
			\begin{axis}[bar shift=5pt,
			legend style={font=\Huge},
			legend style={at={(0.85,0.98)}},
			legend style={draw=none},
			xlabel = \#tuples(M),
			ylabel = Runtime (sec),
			label style = {font=\Huge},
			y label style = {at={(axis description cs:.02,.5)}}]
			\addplot+[draw=none,fill=black!50!white,
			postaction={pattern=north east lines}]
			table [x=C,y=RMA+BAT]{BS_add_ALLSYS.dat};
			\addlegendentry{RMA+BAT};
			\end{axis}
			
			\end{tikzpicture}
			\caption{RMA+BAT vs RMA+MKL}
			\label{fig:ADD_BATvsMKL}
		\end{subfigure}
		\vspace{-5pt}
		\caption{Trip Count (Matrix Addition)}\label{fig:ADD}\vspace{-15pt}
	\end{figure}
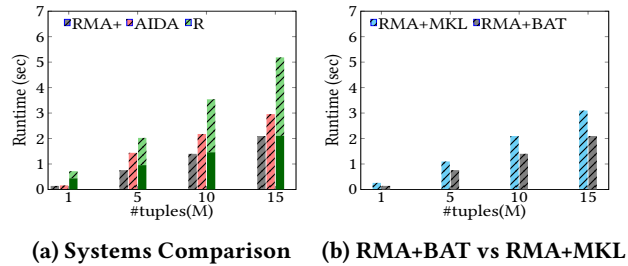
	
	\rev{
		\paragraph{RMA+BAT vs. RMA+MKL} Following our policy, RMA+ delegates
		matrix operations to MKL (RMA+MKL) in
		Figures~\ref{fig:LinReg_systems}, \ref{fig:MulReg_systems}, and
		\ref{fig:CovConf_systems} (the operations are complex and we do not
		run out of memory), and uses the no-copy implementation on BATs
		(RMA+BAT) in Figure~\ref{fig:ADD_systems} ($\add$ is a linear
		operation).  We compare RMA+BAT to RMA+MKL in all scenarios. RMA+MKL
		outperforms RMA+BAT for the queries on trips (factor 1.8-3.8, cf.\
		Figure~\ref{fig:LinReg_BATvsMKL}) and journeys (factor 1.4-1.9, cf.\
		Figure~\ref{fig:MulReg_BATvsMKL}). For the conference query, RMA+MKL
		is 24 to 70 times faster since the cross product requires single
		element access and operates on relations with a large number of
		attributes. For the trip count, RMA+BAT outperforms RMA+MKL in all
		settings (cf.\ Figure~\ref{fig:ADD_BATvsMKL}). Although elementwise
		addition is highly efficient in MKL, the transformation overhead
		cannot be amortized.
	}
	
	\subsection{Discussion}
	
	The key learnings from our empirical evaluation are the following: (1)
	RMA+ excels for mixed workloads that include both standard relational
	and matrix operations. (2) Only RMA+ can avoid data transformations in
	mixed workloads; data transformations may be costly and consume more
	than 90\% of the overall runtime. (3) For complex matrix operations,
	however, transforming the data to a suitable format may pay off: In
	our approach, we are free to transform the data whenever
	beneficial. (4) In terms of scalability to large relations/matrices,
	our solution outperforms all competitors since it relies on the memory
	management of the database system for both the standard relational and
	the matrix operations. (5) Finally, the handling of contextual
	information, a feature of RMA, is efficient and can leverage
	optimizations that avoid expensive sortings.
	
	\section{Conclusion}
	\label{sec:Summary}
	
	\finrev{In this paper, we targeted applications that store data in
		relations and must deal with queries that mix relational and linear
		algebra operations.} We proposed the \emph{relational matrix algebra} (RMA),
	an extension of the relational model with matrix operations that
	maintain important contextual information.  RMA operations are defined
	over relations and can be nested.  We implemented RMA over the
	internal data structures of \mbox{MonetDB} and do not require changes
	in the query processing pipeline. Our integration is competitive with
	state-of-the-art approaches and excels for mixed workloads.
	
	RMA opens new opportunities for cross algebra optimizations that
	involve both relational and linear algebra operations.  It also is
	interesting to investigate the handling of wide tables, e.g., by
	storing them as skinny tables that are accessed accordingly or by
	combining operations to avoid the generation of wide intermediate
	tables.
	
	\section*{Acknowledgments} 
	
	\finrev{We thank Joseph Vinish D'silva for providing the source code
		of AIDA and helping out with the system.  We thank the MonetDB team
		for their support.  We thank Roland Kwitt for the discussion about
		application scenarios.  The project was partially supported by the
		Swiss National Science Foundation (SNSF) through project number
		407550\_167177.}
	
	\bibliographystyle{plain}
	\bibliography{paperbib}

\end{document}